\documentclass[journal]{IEEEtran}

\usepackage{xparse,xcolor}
\ExplSyntaxOn
\NewDocumentCommand{\setupbibcolors}{m}
 {
  \cs_set_protected:Npn \bibitem ##1
   {
    \color{ \str_case:nnF { ##1 } { #1 } { black } }
    \heba_bibitem:n { ##1 }
   }
 }
\cs_set_eq:NN \heba_bibitem:n \bibitem
\ExplSyntaxOff
\setupbibcolors{
  {baza3}{white}
  {baza2}{white}
  {baza1}{white}
  {baza4}{white}
   {baza5}{white}
  {baza6}{white}
  {baza7}{white}
  {firmware2}{white}
   {baza2019sharing}{white}
    {baza2019blockchain}{white}
    {baza9}{white}
    {baza13}{white}
    {baza8}{white}
    {baza12}{white}
    {baza10}{white}
    {baza11}{white}
    {yilmaz2019expansion}{white}
    {yilmaz2}{white}
}
\definecolor{orange}{rgb}{0.97,0.99,1}

\usepackage[utf8]{inputenc}
\usepackage[T1]{fontenc}
\usepackage{url}
\usepackage{ifthen}
\usepackage{cite}
\usepackage[cmex10]{amsmath} 
\usepackage{stackengine}
\usepackage{gensymb}
\usepackage{graphics}
\usepackage{mypackage}
\usepackage{mathrsfs}

\usepackage{mathtools}
\usepackage{tabu}
\usepackage{float}
\usepackage{physics}
\usepackage{siunitx}
\usepackage{multicol}
\usepackage{amssymb}
\usepackage[nomain,acronym,shortcuts]{glossaries}

\usepackage{etoolbox}

\newcommand{\sir}{\mathrm{SIR}}
\newcommand{\Pb}{\mathbb{P}}
\newcommand{\Eb}{\mathbb{E}}

\def\delequal{\mathrel{\ensurestackMath{\stackon[1pt]{=}{\scriptscriptstyle\Delta}}}}

 \makeglossaries
\newcommand*{\acro}[3][]{\newacronym[#1]{#2}{#2}{#3}}

\newtheorem{remark}{Remark}
\newtheorem{proposition}{Proposition}

\theoremstyle{approximation}

\DeclareMathOperator*{\R}{\mathbb{R}}
\DeclareMathOperator*{\N}{\mathcal{N}}

\def\delequal{\mathrel{\ensurestackMath{\stackon[1pt]{=}{\scriptstyle\Delta}}}}

\let\mybibitem\bibitem
\renewcommand{\bibitem}[1]{%
  \ifstrequal{#1}{nature}
    {\color{blue}\mybibitem{#1}}
    {\color{black}\mybibitem{#1}}%
}

\acro{OFDM}{orthogonal frequency-division multiplexing}
\acro{OFDMA}{orthogonal frequency-division multiple access}
\acro{MNO}{mobile network operator}
\acro{RA}{resource allocation}
\acro{SC-FDMA}{single carrier frequency division multiple access}
\acro{CR}{cognitive radio}
\acro{RFIC}{radio frequency integrated circuit} 
\acro{SDR}{software defined radio}
\acro{SDN}{software defined networking}
\acro{su}{secondary user}
\acro{QoS}{quality-of-service}
\acro{USRP}{universal software radio peripheral}
\acro{GSM}{Global System for Mobile Communications}
\acro{TDMA}{time-division multiple access}
\acro{FDMA}{frequency-division multiple access}
\acro{GPRS}{General Packet Radio Service}
\acro{MSC}{Mobile Switching Centre}
\acro{BSC}{Base Station Controller}
\acro{UMTS}{Universal Mobile Telecommunications System}
\acro{WCDMA}{wide-band code division multiple access}
\acro{CDMA}{code division multiple access}
\acro{LTE}{Long Term Evolution}
\acro{PAPR}{peak-to-average power rating}
\acro{HetNet}{heterogeneous networks}
\acro{PHY}{physical layer}
\acro{MAC}{medium access control}
\acro{AMC}{adaptive modulation and coding}
\acro{MIMO}{multiple-input multiple-output}
\acro{M-MIMO}{massive MIMO}
\acro{RAT}{radio access technology}
\acro{VNI}{visual networking index}
\acro{RB}{resource block}
\acro{UE}{user equipment}
\acro{CQI}{channel quality indicator}
\acro{HD}{half duplex}
\acro{IBFD}{in-band full duplex}
\acro{SIC}{self-interference cancellation}
\acro{SI}{self-interference}
\acro{BS}{base station}
\acro{FBMC}{filter bank multi-carrier}
\acro{UFMC}{universal filtered multi-carrier}
\acro{SCM}{single carrier modulation}
\acro{isi}{inter-symbol interference}
\acro{FTN}{faster-than-nyquist}
\acro{M2M}{machine-to-machine}
\acro{MTC}{machine type communication}
\acro{mmWave}{millimeter wave}
\acro{BF}{beamforming}
\acro{LoS}{line-of-sight}
\acro{NLoS}{non-line-of-sight}
\acro{CAPEX}{capital expenditure}
\acro{OPEX}{operational expenditure}
\acro{ICT}{information and communications technology}
\acro{SP}{service providers}
\acro{InP}{infrastructure providers}
\acro{MVNP}{mobile virtual network provider}
\acro{MVNO}{mobile virtual network operator}
\acro{NFV}{network function virtualization}
\acro{VNF}{virtual network functions}
\acro{C-RAN}{cloud radio access network}
\acro{RAN}{radio access network}
\acro{BBU}{baseband unit}
\acro{RRH}{remote radio head}
\acro{TDD}{time-division duplexing}
\acro{FDD}{frequency-division duplexing}
\acro{GFDM}{generalized frequency division multiplexing}
\acro{CSI}{channel state information}
\acro{FFT}{fast Fourier transform}
\acro{IFFT}{inverse FFT}
\acro{CFO}{carrier frequency offset}
\acro{CoMP}{coordinated multipoint}
\acro{D2D}{device-to-device}
\acro{OOB}{out-of-band}
\acro{TTI}{transmission time interval}
\acro{DUE}{D2D user equipment}
\acro{DAS}{distributed antenna system}
\acro{ICIC}{inter-cell interference coordination}
\acro{ICI}{inter-cell interference}
\acro{ISI}{inter-symbol interference}
\acro{CP}{cyclic prefix}
\acro{PDF}{probability distribution function}
\acro{KPI}{key performance indicator}
\acro{SBS}{small base station}
\acro{MBS}{macro base station}
\acro{SCN}{small cell network}
\acro{FIFO}{first in first out}
\acro{VCC}{virtual cache center}
\acro{UAV}{unmanned aerial vehicles}
\acro{MPSQ}{multiclass processor sharing queue} 
\acro{EE}{energy efficiency}
\acro{SIR}{signal-to-interference ratio}
\acro{SINR}{signal-to-noise-plus-interference ratio}
\acro{PPP}{Poisson point process}
\acro{PCP}{Poisson cluster process}
\acro{HCP}{hard-core placement} 
\acro{TCP}{Thomas cluster process }
\acro{CPF}{caching popular files}
\acro{GCA}{greedy caching algorithm }
\acro{RC}{random caching }
\acro{PC}{probabilistic caching}
\acro{5G}{fifth generation}
\acro{MEC}{mobile edge computing}
\acro{AP}{access point}
\acro{VoD}{video-on-demand}
\acro{EPC}{evolved packet core}
\acro{QoE}{quality-of-experience}
\acro{CDN}{content delivery networks}
\acro{F-RAN}{fog-radio access network}
\acro{AR}{augmented reality}
\acro{VR}{virtual reality}
\acro{4C}{computing, caching, communication, and control}
\acro{ABR}{Adaptive BitRate}
\acro{ILP}{Integer Linear Program}
\acro{MILP}{Mixed Integer Linear Program}
\acro{MINLP}{Mixed Integer Non-Linear Program}
\acro{BC}{broadcast channel}
\acro{MDS}{maximum distance separable}
\acro{RS}{Reed-Solomon}
\acro{NC}{network coding}
\acro{MSR}{minimum storage regenerating}
\acro{Coop-MIMO}{cooperative \ac{mimo}}
\acro{UDN}{ultra-dense networks}
\acro{ETSI}{European Telecommunications Standards Institute}
\acro{MCP}{Matern cluster process}
\acro{MD-CoMP}{macrodiversity CoMP transmission}
\acro{JT-CoMP}{joint transmission CoMP}
\acro{CoMP-JT}{coordinated multipoint joint transmission}
\acro{MDSD}{multiple devices to the single device}
\acro{PMF}{probability mass function}
\acro{RV}{random variable}
\acro{i.i.d.}{independently and identically distributed}
\acro{MBMS}{multimedia broadcasting multicasting service}
\acro{CCDF}{complementary cumulative distribution function}
\acro{CDF}{cumulative distribution function}
\acro{PGFL}{probability generating functional}
\acro{KKT}{Karush-Kuhn-Tucker}
\acro{PGF}{point generating function}
\acro{BPP}{binomial point process}
\acro{3GPP}{3rd Generation Partnership Project}
\acro{FAA}{Federal Aviation Administration}
\acro{UAS}{unmanned aircraft systems}
\acro{IoT}{Internet-of-things}
\acro{CNPC}{control and non-payload communication}
\acro{FHD}{full high definition}
\acro{MGF}{moment generating function}
\acro{3D}{three-dimensional}
\acro{2D}{two-dimensional}
\acro{1D}{one-dimensional}
\acro{CB}{conjugate beamforming}
\acro{AU}{aerial user}
\acro{GU}{ground user}
\acro{CLT}{central limit theorem}
\acro{SE}{spectral efficiency}
\acro{MRT}{maximum ratio transmission}
\acro{ZFBF}{zero-forcing beamforming}
\acro{GBS}{ground base station}
\acro{C&C}{command and control}
\acro{ASE}{area spectral efficiency}
\acro{RWP}{random waypoint}
\acro{AGL}{above ground level}
\acro{w.r.t.}{with respect to}
\acro{GTA}{ground-to-air}
\acro{AUE}{aerial user equipment}
\acro{GUE}{ground user equipment}
\acro{UB}{upper bound}
\acro{LB}{lower bound}
\acro{ABS}{aerial base station}
\acro{UAV-UE}{UAV-\ac{UE}}
\acro{SCDP}{successful content delivery probability}
\acro{HARP}{highest average received power}
\acro{ULA}{uniform linear array}
\acro{GPP}{Gaussian Poisson process}
\acro{NSD}{nearest serving device}
\acro{NCP}{nearest content provider}
\acro{RSD}{randomly serving device}
\acro{RSCP}{randomly-selected content provider}
\acro{BCD}{block coordinate descent}
\acro{NOMA}{non-orthogonal multiple access}
\acro{MPC}{most popular content}
\acro{LPC}{least popular content}
\acro{AI}{artificial intelligence}
\acro{ML}{machine learning}
\acro{TPP}{temporal point process}
\acro{PP}{point process}
\acro{DPP}{Determinantal point process}
\acro{DPPL}{DPP-based learning}
\acro{DRL}{deep reinforcement learning}
\acro{RL}{reinforcement learning}
\acro{RNN}{recurrent neural network}
\acro{RMTPP}{recurrent marked TPP}
\acro{NR}{New Radio}
\acro{HO}{handover}
\acro{RLF}{radio link failure}
\acro{RSRP}{reference signal received power}

\begin{document}
\title{Optimized Caching and Spectrum Partitioning for D2D enabled Cellular Systems with Clustered Devices}
 
\author{Ramy Amer,~\IEEEmembership{Student~Member,~IEEE,} Hesham~Elsawy,~\IEEEmembership{Senior~Member,~IEEE,} M.~Majid~Butt,~\IEEEmembership{Senior~Member,~IEEE,} Eduard~A.~Jorswieck,~\IEEEmembership{Fellow IEEE,}
Mehdi~Bennis,~\IEEEmembership{Senior~Member,~IEEE,}~and~Nicola~Marchetti,~\IEEEmembership{Senior~Member,~IEEE}
\thanks{Ramy Amer and Nicola~Marchetti are with CONNECT Centre for Future Networks, Trinity College Dublin, Ireland. Email:\{ramyr, nicola.marchetti\}@tcd.ie.}
\thanks{Hesham ElSawy is with King Fahd University of Petroleum and Minerals (KFUPM), Saudi Arabia. Email: hesham.elsawy@kfupm.edu.sa.}
\thanks{M. Majid Butt is with Nokia Bell Labs, France, and CONNECT Centre for Future Networks, Trinity College Dublin, Ireland. Email: Majid.Butt@tcd.ie.}
\thanks{Eduard A. Jorswieck is with Institute for Communications Technology
TU Braunschweig, Germany, Email: jorswieck@ifn.ing.tu-bs.de.}
\thanks{Mehdi Bennis is with the Centre for Wireless Communications, University of Oulu, Finland, and the Department of Computer Engineering, Kyung Hee University, South Korea. Email: mehdi.bennis@oulu.fi.}
\thanks{This publication has emanated from research conducted with the financial support of Science Foundation Ireland (SFI) and is co-funded under the European Regional Development Fund under Grant Number 13/RC/2077.}}
\maketitle

%
%
%

\markboth{
 Accepted for Publication in IEEE TRANSACTIONS ON COMMUNICATIONS, VOL. ?, NO. ?, 2020.}
 {Shell \MakeLowercase{\textit{et al.}}: Bare Demo of IEEEtran.cls for Computer Society Journals}

\begin{abstract} 
Caching at mobile devices and leveraging device-to-device (D2D) communication are two promising approaches to support massive content delivery over wireless networks. The analysis of cache-enabled wireless networks is usually carried out by assuming that devices are uniformly distributed, however, in social networks, mobile devices are intrinsically grouped into disjoint clusters. In this regards, this paper proposes a spatiotemporal mathematical model that tracks the service requests arrivals and account for the clustered devices geometry. Two kinds of devices are assumed, particularly, content clients and content providers. Content providers are assumed to have a surplus memory which is exploited to proactively cache contents from a known library, following a random probabilistic caching scheme. Content clients can retrieve a requested content from the nearest content provider in their proximity (cluster), or, as a last resort, the base station (BS). The developed spatiotemporal model is leveraged to formulate a joint optimization problem of the content caching and spectrum partitioning in order to minimize the average service delay. Due to the high complexity of the optimization problem, the caching and spectrum partitioning problems are decoupled and solved iteratively using the block coordinate descent (BCD) optimization technique. To this end, an optimal and suboptimal solutions are obtained for the bandwidth partitioning and probabilistic caching subproblems, respectively. Numerical results highlight the superiority of the proposed scheme over conventional caching schemes under equal and optimized bandwidth allocations. Particularly, it is shown that the average service delay is reduced by nearly $100\%$ and $350\%$, compared to the Zipf and uniform caching schemes under equal bandwidth allocations, respectively.		
\end{abstract}
\vspace{-0.0 cm}
\begin{IEEEkeywords}
D2D communication, spatiotemporal, probabilistic caching, delay analysis, queuing theory.
\end{IEEEkeywords}
\vspace{-0.0 cm}
\section{Introduction}
Caching at mobile devices significantly improves system performance by facilitating \ac{D2D} communications, which enhances the spectrum efficiency and alleviates the heavy burden on backhaul links \cite{amer20200caching,8412262, 6600983}. Prior works in the literature commonly followed two approaches for the content placement, namely, deterministic placement and probabilistic placement. For deterministic placement, contents are cached and optimized for specific networks in a deterministic manner \cite{6600983,6787081,8412262,Delay-Analysis}. However, in practice, the wireless channels and the geographic distribution of devices are time-variant. This requires frequent updates to the optimal content placement strategy, which is a highly complicated task. To cope with this problem, probabilistic content placement considers that each device randomly caches a subset of content with a certain caching probability \cite{chen2017probabilistic,7248843,7502130}. 
In this paper, we particularly focus on probabilistic content placement. 

Modeling cache-enabled \acp{HetNet}, including \acp{BS} and mobile devices, followed two main directions in the literature. The first line of work, popularly named as protocol model, focuses on the fundamental throughput scaling results by assuming a simple channel model \cite{6600983,6787081,8412262}. In this model, two devices can communicate if they are within a certain distance. The second line of work, defined as the physical interference model, considers a more realistic model for the underlying physical layer \cite{7931641,cache_schedule,8647532,amer2020caching,amer2019performance,8886101}. This physical interference model is based on the fundamental \ac{SIR} metric, and hence, is applicable to any wireless communication system whose performance is measured by some \ac{SIR}-based utility function. In this work, we consider a realistic physical layer model where all transmissions are subject to outage due to fading and interference. Next, we review relevant literature to the \ac{D2D} caching networks.


\vspace{-0.2 cm}
\subsection{State of the Art and Prior Works}
Modeling and analysis of \ac{D2D} communication networks are widely adopted in the literature with the assumption  that the mobile devices are uniformly distributed, especially in the wireless caching area \cite{7931641} and \cite{cache_schedule}. For instance, the authors in \cite{7931641} investigated the relationship between offloading probability of a cache-enabled \ac{D2D} network and the energy cost of caching helpers that are uniformly distributed according to a \ac{PPP}. Meanwhile, the authors in \cite{cache_schedule} jointly optimized caching and scheduling to maximize the offloading probability of a network whose users' locations follow a \ac{PPP}. For such \ac{D2D} caching networks, \ac{PPP} is an appealing analytical framework due to its simplicity and tractability \cite{6042301} and \cite{haenggi2012stochastic}. However, a realistic model for \ac{D2D} caching networks needs to capture the notion of device clustering, which is fundamental to the D2D network architecture \cite{zhang2015social} and \cite{hu2014evaluating}. Clustered \ac{D2D} models imply that a given device typically has multiple proximate devices, where any of them can act as a serving device. Such models can be effectively described by cluster processes \cite{clustered_twc,clustered_tcom,8374852,8536464}. 

In this regard, the authors in \cite{clustered_twc} developed a stochastic geometry-based model to characterize the performance of content placement in a clustered \ac{D2D} network. In particular, the authors proposed different  strategies of content placement in a \ac{PCP}  deployment. Meanwhile, the work in \cite{clustered_tcom} proposed a cluster-centric content placement scheme where the content of interest is cached closer to the cluster center. Moreover, the authors in \cite{8374852} proposed cooperation among the \ac{D2D} transmitters and hybrid caching strategies to save the energy cost of content providers, where the location of these providers is modeled by a \ac{GPP}. Inspired by the Matern hard-core point process, the authors in \cite{8070464} proposed a spatially correlated caching strategy for which devices that cache the same content can not get closer to each other. However, while interesting, the works in \cite{clustered_twc} and \cite{clustered_tcom} only characterized the performance of clustered \ac{D2D} networks with the assumption that contents are pre-cached, i.e., there was no study of the caching problem. Moreover, \cite{8374852} and  \cite{8070464} focused solely  on optimizing the content placement for \ac{D2D} networks to maximize the offloading probability. In particular, there was no study of joint caching and communication, e.g., bandwidth allocation and spectrum access, which is vital to improve important performance metrics, e.g., service delay and spectral efficiency.


%
%
%

Recently, delay analysis and minimization for wireless caching networks have received significant attention, see, e.g.,  \cite{6600983,amer20200caching,8412262} and \cite{yang2016analysis,huang2018energy,7797148}. For instance, \cite{8412262} proposed an inter-cluster collaborative caching architecture to reduce the average service delay for clustered \ac{D2D} networks. Meanwhile, the authors in \cite{6600983} advocated to set up helpers with caching capability in cellular networks to reduce the access delay. The authors in \cite{yang2016analysis} proposed a cache-based content delivery scheme in a three-tier \ac{HetNet} consisting of \acp{BS}, relays, and \ac{D2D} pairs. Particularly, the average throughput and delay are characterized based on a queuing model and continuous-time Markov process. In \cite{huang2018energy}, the authors maximized the energy efficiency in an additive white Gaussian channel with content caching subject to delay constraint. The authors in \cite{7797148} proposed joint request routing and content caching to minimize the average access delay. However, in these works, the authors   studied the content placement problem to analyze and minimize the network delay while the joint optimization of caching and communication was not considered.

Compared with this prior art \cite{clustered_twc,clustered_tcom,8536464,8374852,8070464,yang2016analysis,huang2018energy,7797148},  this paper provides a comprehensive performance analysis and optimization for a clustered \ac{D2D} caching network. In particular, we first develop a spatiotemporal model that accounts for the requests' arrival, the clustered geometry of devices, content caching, and overlay interactions between the cellular and D2D communications. Then, we formulate an optimization problem based on the developed spatiotemporal model aiming to reduce the average service delay. Our approach effectively captures the notion of device clustering and accounts for resource allocation and scheduling of devices. These aspects have not been addressed yet in the literature, especially in the context of the joint design of caching and spectrum partitioning. \emph{To the best of our knowledge, this is the first work to propose a joint caching and spectrum partitioning scheme to reduce the average service delay for clustered D2D networks. Moreover, this paper introduces the first spatiotemporal analysis for D2D cache-enabled networks.} 
The main contributions of this paper are summarized as follows:			
\begin{itemize}
\item We consider a spatially clustered \ac{D2D} caching network in which users with different interests exist in disjoint vicinities, i.e., different popularity profiles per clusters. For this network, we study the content caching and delivery among clusters' devices and propose a joint caching and spectrum partitioning scheme to reduce the average service delay.
\item We develop a spatiotemporal model by combining tools from stochastic geometry and queuing theory, which allows us to account for the notion of device clustering, requests' arrival and service rates, and the  traffic queue dynamics. We then conduct the rate coverage analysis to obtain the service rates of the traffic modeling queues. We can then characterize the conditions under which the network sustains its stability, and obtain the request average service delay as a function of the system parameters. 
\item Towards a minimized service delay, we jointly optimize spectrum partitioning between \ac{D2D} and \ac{BS}-to-device communications and content caching. Given the non-convexity of the joint problem, the bandwidth allocation and caching subproblems are decoupled and solved iteratively using \ac{BCD} optimization technique. In particular, we characterize the optimal bandwidth allocation in a closed-form expression, and a suboptimal solution for the content caching subproblem is also obtained.  
\item Our results reveal that the optimal bandwidth allocation heavily depends on the request arrival rate, popularity of files, and the network geometry. Moreover, it is shown that the average service delay can be significantly reduced by the joint optimization of spectrum partitioning and content caching compared to other benchmark schemes.
\end{itemize}
The rest of this paper is organized as follows. Section II and Section III present, respectively, the system model and  rate coverage probability analysis. The average service delay minimization is then conducted in Section IV. Numerical results are presented in Section V and conclusions are drawn in Section VI. 
\begin{figure}[!t]
\vspace{-0.6 cm}
\centering
\includegraphics[width=0.40 \textwidth]{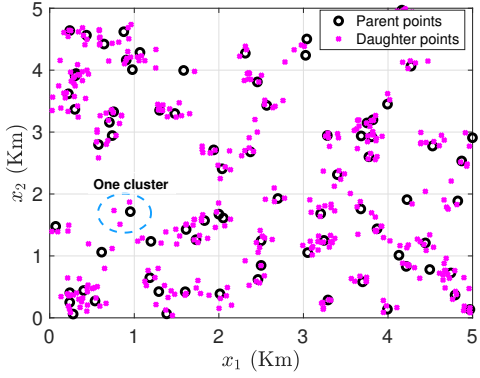}		
\vspace{-0.4cm}			
\caption {A realization of the proposed TCP network whose parent points form PPP $\Phi_p$ of density $\lambda_p=5$ \SI{}{km}$^{-2}$, and the average number of devices per cluster is $\bar{n}=5$. The standard deviation of normal scattering (distribution) is $\sigma=\SI{100}{m}$.} 
\label{sys-model}
\vspace{-0.5 cm}
\end{figure}
\vspace{-0.4 cm}
\section{System Model}	

\subsection{Network Model}
\label{netw-model}
We consider a clustered \ac{D2D} cache-enabled network in which devices can share their cached content within the same cluster. For this network, we model the location of the devices with a \ac{TCP} composed of parent and their corresponding daughter points. A general \ac{TCP} is generated by taking a parent homogeneous \ac{PPP} and daughter Gaussian PPP, one per parent, and translating the daughter processes to the position of their parents \cite{haenggi2012stochastic}. The cluster process is then the union of all the daughter points. Let us denote the parent point process by $\Phi_p = \{\boldsymbol{x}_1,\boldsymbol{x}_2,\dots\}$, where $\boldsymbol{x}_i\in\R^2$, and $i\in\N$, see Fig.~\ref{sys-model}. Further, let $(\Phi_i)$ be a family of finite point sets representing the untranslated daughter Gaussian \acp{PPP}, i.e., untranslated clusters. The cluster process is then the union of the translated clusters: 
\begin{align}
\Phi \delequal \cup_{i\in \N} \boldsymbol{x}_i + \Phi_i.
\end{align}
The parent points and offspring are referred to as cluster centers and cluster members, respectively. We assume that the cluster centers are drawn from the PPP $\Phi_p$ whose density is $\lambda_p$, while the number of cluster members is a Poisson \ac{RV} with a certain mean. We also assume that, for Gaussian \acp{PPP}, the cluster members (daughter points) are normally scattered of variance $\sigma^2 \in \mathbb{R}$ around their cluster centers (parent points) \cite{haenggi2012stochastic}. Given this normal scattering of daughter points, the \ac{PDF} of the cluster member location relative to its cluster center equals

\begin{equation}
f_{\boldsymbol{Y}}(\boldsymbol{y}) = \frac{1}{2\pi\sigma^2}\textrm{exp}\Big(-\frac{\lVert \boldsymbol{y}\rVert^2}{2\sigma^2}\Big),	\quad \quad  \boldsymbol{y} \in \mathbb{R}^2,
\label{pcp}
\end{equation}
where $\boldsymbol{y}\in\R^2$ is the device location relative to its cluster center, $\lVert .\rVert$ is the Euclidean norm. If the average number of devices per cluster is $\bar{n}$, the cluster intensity function will be:

\begin{align}
\lambda_c(\boldsymbol{y}) = \frac{\bar{n}}{2\pi\sigma^2}\textrm{exp}\big(-\frac{\lVert \boldsymbol{y}\rVert^2}{2\sigma^2}\big),	\quad \quad  \boldsymbol{y} \in \mathbb{R}^2. 
\end{align}
Therefore, the intensity of the entire process $\Phi$ will be $\lambda = \bar{n}\lambda_p$. We assume two kind of devices co-exist within the same cluster, namely, content clients and content providers as done in   \cite{8374852}. In particular, the devices that can perform proactive caching and provide content delivery are called content providers while those requesting content are called content clients. All content providers are assumed to have the same transmission power $P_d$.

%

We also consider a tier of \acp{BS}, which are connected to the core network and communicate with the content clients only when \emph{their requests cannot be satisfied by D2D links}. We assume that all the content that might be requested by content clients are available at the \acp{BS}. The \acp{BS}' spatial distribution follows another \ac{PPP} $\Phi_{b}$ with density $\lambda_{b}$, which is independent of $\Phi_p$. All \acp{BS} have the same transmission power $P_b$. We assume an overlay operation such that the BSs and D2D transmitters are assigned non-overlapping frequency bands to avoid cross-tier interference. In particular, the total system bandwidth $W$ is  divided into two portions, $W_b$ for the \ac{BS}-to-device communication, and $W_d=W-W_b$ for D2D communications. The cellular and \ac{D2D} channels are characterized by \ac{NLoS} communications and both large-scale and small-scale fading, e.g., see \cite{clustered_twc,clustered_tcom,8374852,8070464,8536464}. Large-scale fading is represented by the power path loss law $\ell(w)=w^{-\alpha}$, where $w$ is the communication distance and $\alpha>2$ is the path loss exponent. Moreover, the small-scale fading follows a Rayleigh distribution whose channel gain $g_w$ is exponentially distributed with unit mean, i.e., $g_w \sim{\rm exp}(1)$.\footnote{In Section \ref{num-result}, the channel model is revisited to evaluate the performance of a clustered cache-enabled \ac{D2D} network in which intra-cluster devices experience \ac{LoS} communications.} We set all transmit powers to unity and focus on the interference-limited regime, thus, the thermal noise is ignored.

In the adopted system model, we assume that each transmitter, either BS or device, sends codewords drawn from a Gaussian codebook with a fixed rate ${\rm log}(1+\theta)\Pb_c$, where $\theta$ is the target \ac{SIR} threshold and $\Pb_c$ is the rate coverage probability. This fixed rate expression will be explained later in our discussions. A quasi-static channel model is adopted where the channel gain is assumed to be constant during one codeword transmission, and changes from one codeword to another \cite{7733098}. The codeword length is determined based on the underlying coding and modulation schemes, which are captured by the fixed rate expression ${\rm log}(1+\theta)\Pb_c$. Next, we introduce our proposed content caching and delivery schemes.

\vspace{-0.6 cm}
\subsection{Content Popularity and Probabilistic Caching Model}
We assume that each content provider has a surplus memory of size $M$ files, designated for caching contents. 
The total number of contents is $N_f\geq M$ and the set (library) of content indices is denoted as $\mathcal{F} = \{1, 2, \dots , N_f\}$. These contents represent the content catalog that all the devices in a cluster may request, which are indexed in a descending order of popularity. The probability that the $i$-th content is requested follows a Zipf distribution given by \cite{749260},
\begin{equation}
q_i = \frac{ i^{-\beta} }{\sum_{k=1}^{N_f}k^{-\beta}},
\label{zipf}
\end{equation}
where $\beta$ is a parameter that reflects how skewed the popularity distribution is. For example, if $\beta= 0$, the popularity of the contents has a uniform distribution. Increasing $\beta$ increases the disparity among the contents' popularity such that lower indexed contents have higher popularity. By definition, $\sum_{i=1}^{N_f}q_i = 1$. We assume that the popularity of files might differ among different clusters. This is motivated by the fact that disjoint user groups in different locations may have different interests. Hence, we use the Zipf distribution to model the popularity of contents per cluster. The popularity of contents is assumed to be perfectly known.\footnote{Given the time-varying content popularity in practical scenarios, incorporation of estimation errors of content popularity might be required to convey a more conservative study of the network performance \cite{7502130}. As an example of factoring in the time-varying popularity, the classical web caching systems adopt dynamic eviction policies like least recently used in order to combat time-varying content popularity in a heuristic manner \cite{7537172}. However, considering time-varying popularity is out of the scope of this work. For a detailed review of the main results and the literature on the topic of time-varying content popularity, the reader is referred to \cite{8374867,7524380,8063333}.}


\begin{figure}[!tbh]
\vspace{-0.0 cm}
	\begin{center}
		\includegraphics[width=2.2 in]{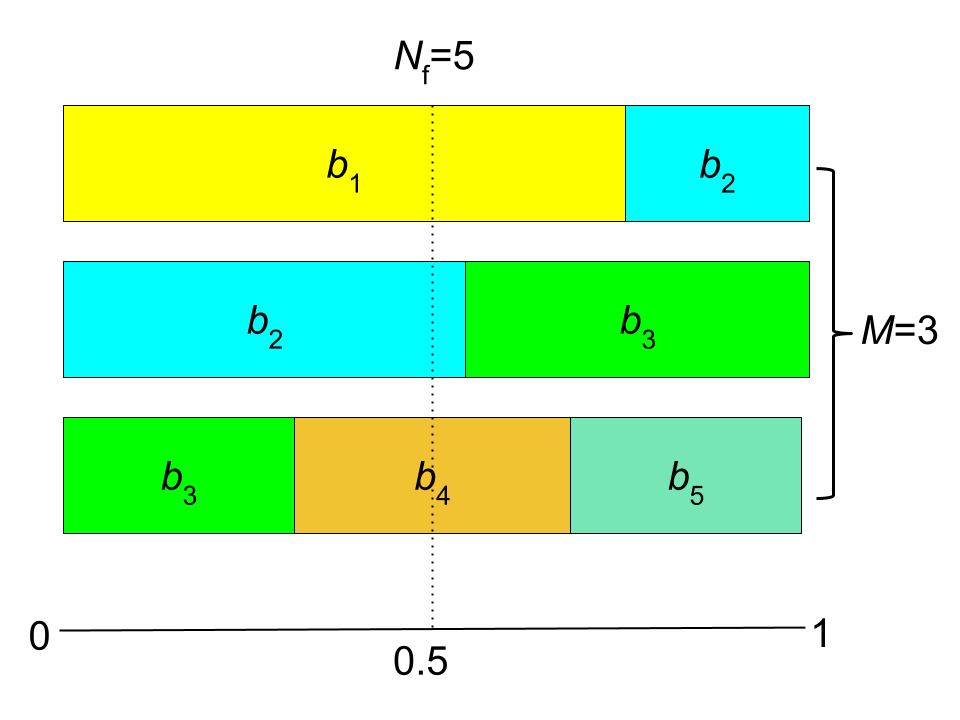} 
		\caption {Illustrative example of the proposed \ac{PC} scheme for a library of files of size $N_f=5$ and maximum cache size per device of $M=3$ files.}
		\label{prob_cache_example}
	\end{center}
\vspace{-0.4 cm}
\end{figure}

For the caching policy, we adopt a random content placement scheme where each content provider independently selects a content to cache according to a specific probability function $\boldsymbol{b} = \{b_1, b_2, \dots, b_{N_{f}}\}$, where $b_i$ is the probability that a device caches the $i$-th content, $0 \leq b_i \leq 1, \forall i\in\mathcal{F}$. To avoid duplicate caching of the same content within the memory of the same device, we follow a probabilistic caching approach proposed in \cite{7248843} and illustrated in Fig.~\ref{prob_cache_example}. An illustrative example of the adopted \ac{PC} is given in Fig.~\ref{prob_cache_example}, for a library size $N_f=5$ and per-device cache size $M = 3$. Firstly, $M$ is equally divided into $3$ (vertical) blocks of unit size. Then, for given caching probabilities $\{b_1,b_2,\dots,b_5\}$, with $\sum_{i=1}^{N_f=5}b_i = M=3$, the three vertical blocks are populated with the $b_i$ values. Finally, a random number $\in [0,1]$ is generated and content $i$ is chosen from each block whose $b_i$ fills the partition  intersecting with the generated random number. In this way, if the random number is 0.5 for the given example in Fig.~\ref{prob_cache_example}, the contents $\{1, 2, 4\}$ are chosen to be cached. The caching probabilities $b_i$, $i\in\mathcal{F}$, can be tuned so as to optimize a particular performance metric, e.g., the average service delay as detailed later. It is worth mentioning that \ac{PC} is a standard caching technique that is widely adopted in the literature, see, e.g., \cite{chen2017probabilistic,7248843,7502130}. Next, we will discuss the content sharing via D2D communications and explain the mechanism of D2D link setup.

\vspace{-0.4 cm}
\subsection{Content Request and Delivery Model} 
Due to the cost of participating in content caching and delivery, e.g., battery consumption and memory utilization, not all content providers can be active in all time slots. Furthermore, some content providers may be busy in their own communications. Hence, within each cluster, we assume that content providers can be available for content delivery with probability $p\in[0,1]$.\footnote{Incentivizes for the devices to cache and deliver contents is a standalone problem that is studied in the literature, see, e.g., \cite{cache_schedule}, however, it is out of our scope in this paper.}



%
%

Within each cluster, we assume a content client device whose distance to its cluster center is drawn from a Rayleigh distribution of scale parameter $\sigma$, according to the \ac{TCP} definition. Throughout time, content clients in different clusters may request files $i\in\mathcal{F}$ with a probability following the assumed per-cluster Zipf distribution in (\ref{zipf}). Similar to \cite{8412262} and \cite{6600983}, we neglect the device self-cache and assume that desired contents can be only brought by \ac{D2D} or \ac{BS}-to-device communications. Since each cluster has its own library, a given content client may either be served via a D2D connection from the nearest active provider within the same cluster or, as a last resort, via the nearest geographical BS. 

For our network, we assume a \ac{BS}-assisted \ac{D2D} link setup that works as follows \cite{8458381}. A content client first sends its request to its geographical closest \ac{BS}, which knows the active content providers, their cashed files, as well as their locations. If there is an active content provider caching the requested file, the BS then establishes a direct D2D link between the content client and its geographically closest active content provider. Otherwise, the content client attaches to this closest BS, as a last resort, to download the desired file. Next, we illustrate the adopted traffic and queueing models.



\begin{figure}
\vspace{-0.0 cm}
	\begin{center}
		\includegraphics[width=3.0 in]{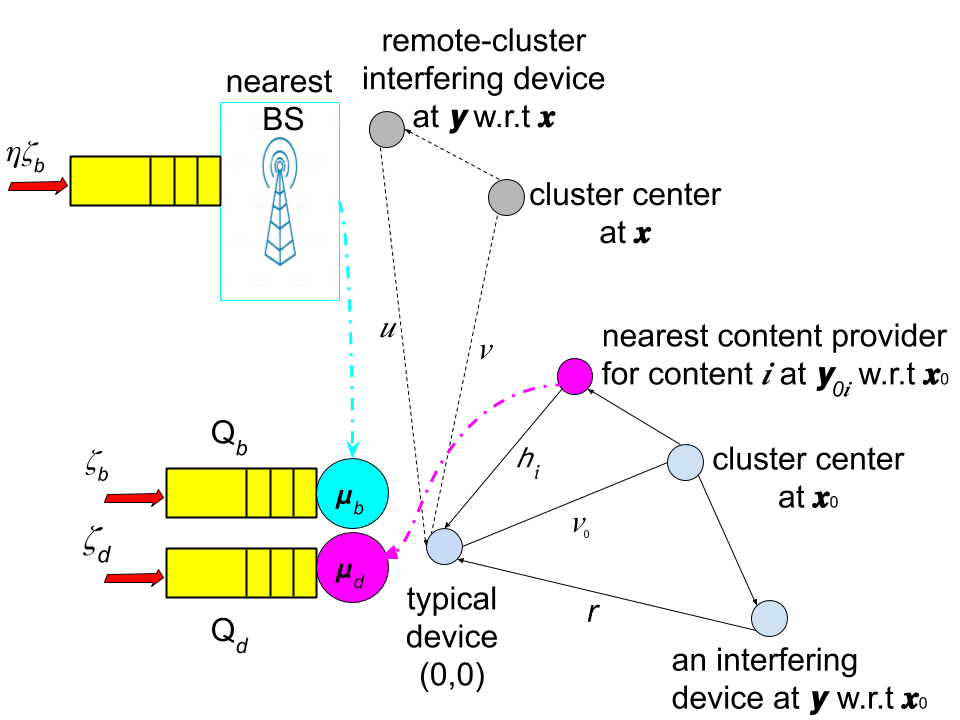}		
		\caption {The traffic and communication models for a content client in a given cluster are shown. Queues Q$_d$ and Q$_b$ are used to model requests served via \ac{D2D} and \ac{BS}-to-device communications, respectively. In the top left, $\eta\zeta_b$ represents the aggregate arrived requests from $\eta$ clients that are served by the same \ac{BS} on the same physical radio resource.}
\label{distance-delay-queue}
	\end{center}
\vspace{-0.4 cm}
\end{figure}
%
\vspace{-0.3 cm}
\subsection{Traffic Model} 
We assume that the per client request arrival follows a Poisson arrival process with parameter $\zeta$ (requests per time slot). As shown in Fig.~\ref{distance-delay-queue}, the incoming requests are further divided according to where they are served from. The arrival rate of requests served via the \ac{D2D} communication is denoted by $\zeta_{d}$, while $\zeta_{b}$ is the arrival rate for those served from the \ac{BS}. By definition, $\zeta_{d}$ and $\zeta_{b}$ are also Poisson arrival processes. Without loss of generality, we assume that the content size has a general distribution $G$ whose mean is denoted as $\overline{S}$ \SI{}{MBytes}. Hence, an M/G/1 queuing model is adopted whereby two non-interacting queues, Q$_d$ and Q$_b$, model the client traffic served via \ac{D2D} and \ac{BS}-to-device communications, respectively. Let $\mu_d$ and $\mu_b$ be the service rates of Q$_d$ and Q$_b$, respectively. Although Q$_d$ and Q$_b$ are non-interacting queues, as the \ac{D2D} communication is assumed to be out-of-band, these two queues are spatially interacting with their similar queues in other clusters (or \acp{BS}), as we will detail later.

We assume that the average number of potentially active devices, i.e., content providers, is $\bar{n}$. Hence, according to the proposed channel access scheme, the set of active devices within the same cluster follows a Poisson \ac{RV} whose mean is $p\bar{n}$. The probability that there are $k$ active devices per cluster is hence equal to $\frac{(p\bar{n})^k e^{-p\bar{n}}}{k!}$. Accordingly, the probability that there exists $k$  active providers for content $i$ is $\frac{(b_ip\bar{n})^k e^{-b_ip\bar{n}}}{k!}$, where $b_i$ is the caching probability. Hence, the probability that there exists at least one active provider for content $i$ is $1$-minus the void probability (i.e., $k=0$), given by $1 - e^{-b_ip\bar{n}}$. It is worth highlighting that Q$_b$ and Q$_d$ account for requests for content demanded by a client, not the content itself. These content requests are assumed of finite-small sizes such that the effect of the limited buffer size of Q$_d$ and Q$_b$ can be ignored. The assumption of finite-small packet sizes of requests and, accordingly, the intangible impact of buffer sizes of Q$_d$ and Q$_b$  is perfectly practical as we deal with requests for content, not the content itself.      
However, if the accumulated packets or requests are of significant sizes compared to the buffer sizes, admission control policies can be imposed. In particular, admission policies could be applied at the device or BS levels such that there is a maximum number of accumulated requests in the device queue or a maximum number of devices to be served by the BS, respectively. This would help limit the queue overflow, i.e., reduce the probability that the number of  accumulated requests exceeds a certain threshold. As such, devices or requests being served undergo finite truncated delay with no probability of overflow at the buffers or experiencing infinite delays. However, this will  happen at the expense of more discarded requests and fewer number of users to be served by the BSs. The spatiotemporal analysis for caching networks with underlying admission policies is beyond the scope of this paper. Nonetheless, it can be a good subject for future work.

Given the Poisson arrival rate $\zeta$ and the Zipf popularity of content, we get
\begin{align}
\label{zeta-d}
\zeta_{d} &= \zeta \sum_{i=1}^{N_f} q_i(1 - e^{-b_ip\bar{n}}),
\\
\label{zeta-b}
\zeta_{b}&= \zeta \sum_{i=1}^{N_f} q_ie^{-b_ip\bar{n}}. 
\end{align}

We assume a \ac{FIFO} scheduling technique whereby a request for content arrives first will be scheduled first either by the \ac{D2D} or \ac{BS}-to-device communication, if the content is cached among the active providers or not, respectively. Notice that the result of FIFO scheduling at BSs only relies on the time when the request arrives at the queue, i.e., it is irrelevant to the particular device that issues the request.

Having obtained the arrival rates, next, we characterize the corresponding service rates of Q$_d$ and Q$_b$. To do that, we will turn our attention to the communication part of our system. We assume a \ac{SIR} threshold model that works as follows: If the \ac{SIR} of  a wireless link is above a threshold $\theta$, the link can be successfully used for information transmission at spectral efficiency ${\rm log}_{2}(1+ \theta)$ $\SI{}{bits/sec/Hz}$. In other words, we adopt a fixed rate transmission scheme, whereby each device (or BS) transmits at the fixed rate of ${\rm log}_2(1+\theta)$ \SI{}{bits/sec/Hz}, where $\theta$ is a design parameter. Note that such definition is widely adopted and accepted in the literature as in \cite{8374852,7733098}, and \cite{zhong2017heterogeneous,zhong2016stability,stamatiou2010random}. Since the rate is fixed, the transmission is prone to outage due to fading and interference fluctuations. Consequently, the de facto average transmissions rate, i.e., average throughput (aka goodput), can be expressed similar to \cite{8374852,7733098}, and \cite{zhong2017heterogeneous,zhong2016stability,stamatiou2010random} as
\begin{equation} 
\label{rate_eqn}
C_l = W_l{\rm log}_{2}(1+ \theta)\Upsilon_l,			
\end{equation}						
where $l\in\{d,b\}$ for the \ac{D2D} and \ac{BS}-to-device communication, respectively, $\Upsilon_l=\mathbb{E}[\boldsymbol{1} \{\sir>\theta\}]$, and $\textbf{1}\{.\}$ is the indicator function. The rate coverage probability $\Upsilon_l$ is defined as the probability that a content is downloadable with an \ac{SIR} higher than a certain threshold $\theta$.


Given the memoryless property of the Poisson arrival process along with the fact that the service process is independent of the arrival process, we reach the following result: The number of requests in each queue at a future time only depends upon the current number in the system (at time $t$) and the arrivals or departures that occur within the interval $e$, i.e., 
\begin{align}
Q_{l}(t+e) =  Q_{l}(t) +   \Lambda_{l}(e)  -   M_{l}(e),  
\end{align}
where $l\in\{d,b\}$, and $\Lambda_{l}(e)$ is the number of arrivals in the time interval $(t,t+e)$ whose mean is 
$\zeta_l$. Similarly,  $M_{l}(e)$ is the number of departures in the same interval, which has the mean
\begin{align}
\label{serv-rate}
\mu_l =\frac{C_l}{\overline{S}} =\frac{W_l{\rm log}_2(1+\theta)\Upsilon_l}{\overline{S}}.  
\end{align}
The mean service time is hence $\tau_l= \frac{1}{\mu_l}=\frac{\overline{S}}{W_l{\rm log}_2(1+\theta)\Upsilon_l}$, which follows the same distribution as the content size. Since $\Upsilon_l=\mathbb{E}[\boldsymbol{1} \{\sir>\theta\}]$, it is clear that the service rates $\mu_l$ depend on the traffic dynamics. This is because a request being served via Q$_d$ (or Q$_b$) in a given cluster is interfered only from other clusters (or \acp{BS}) that serve requests via \ac{D2D} (or \ac{BS}-to-device) communications. More precisely, lower encountered request arrival rates are equivalent to lower interference received at the client device, and vice versa. This observation will be revisited later in our discussions in Section \ref{del-anal}. Having described the underlying spatial and temporal models, next, we introduce our main performance metric, namely, the average service delay. We particularly relate this metric to the underlying spatiotemporal model and explain the proposed joint caching and bandwidth partitioning optimization problem.

\subsection{Minimized Average Service Delay}
Serving a client by fetching its desired content via either D2D or BS-to-device communication is susceptible to  latency. This is because requests are subject to both queuing and transmission delays. Ensuring low latency for applications such as video streaming or \ac{VR} is a critical matter for the \ac{QoE} perceived by the users \cite{chaccour2019Reliability}. In particular, human eyes need to experience smooth sequences of video frames with sufficiently low latency. However, temporary outages due to impairments in the received \ac{SINR} are frequent in wireless environments. Hence, the joint design of caching and communication, e.g., bandwidth allocation, is essential to ensure reduced service delay and improved \ac{QoE}. Therefore, our target is to  study and minimize the per request average service delay incurred to deliver contents to the corresponding devices (clients). 

From our previous discussions on the queuing and communication models, we first observe that the caching probability $b_i$ controls the arrival rates $\zeta_d$ and $\zeta_b$ of Q$_d$ and  Q$_d$, respectively (from (\ref{zeta-d}) and (\ref{zeta-b})). Moreover, the bandwidths $W_d$ and $W_b$ determine their corresponding service rates, namely, $\mu_d$ and $\mu_b$. For instance, if $W_d$ is arbitrarily larger than $W_b$, requests served via the  D2D communication mode might experience lower delay than those served via the cellular communication mode, and vice versa. This, in turn, could yield overall larger weighted average service delay. Moreover, allocating bandwidths to  these two communication modes regardless of their percentage of served traffic will lead to inefficient spectrum utilization. Hence, it is vital to jointly design the caching scheme and allocate the bandwidth between the \ac{D2D} and cellular communications such that the average service delay is efficiently reduced. Finally, to obtain the service rates $\mu_l$ in (\ref{serv-rate}), $l\in\{d,b\}$, we must calculate the rate coverage probabilities $\Upsilon_l$. We next conduct the rate coverage analysis for both \ac{D2D} and cellular communications to obtain the rate coverage probabilities $\Upsilon_l$.		


\vspace{-0.0 cm}
\section{Rate Coverage Probability}  

\subsection{D2D Coverage Probability}
Without loss of generality, we conduct the next analysis for a cluster whose center is at $\boldsymbol{x}_0\in \Phi_p$, referred to as \emph{representative cluster}, and a client device, henceforth called \emph{typical device}, located at the origin. We denote the location of the nearest active provider for content $i$ by $\boldsymbol{y}_{0i}$   relative to the cluster center $\boldsymbol{x}_0$. From Fig. \ref{distance-delay-queue}, the distance from the typical  device to this provider is $h_i=\lVert \boldsymbol{x}_0+\boldsymbol{y}_{0i}\rVert$. We first characterize the \ac{PDF} $f_{H_i}(h_i)$, then, we readily obtain the rate coverage probability and transmission rate.

That being said, that average number of content providers per cluster is set to be $\bar{n}$. Hence, active content providers for file $i\in\mathcal{F}$ form a Gaussian \ac{PPP} $\Phi_{ci}$ of intensity $\lambda_{ci}(y) = b_i p \lambda_c(y)= \frac{b_ip\bar{n}}{2\pi\sigma^2}\textrm{exp}\big(-\frac{\lVert \boldsymbol{y}\rVert^2}{2\sigma^2}\big)$ by the thinning theorem \cite{haenggi2012stochastic}.  Based on that, the \ac{PDF} of the distance to the nearest active provider for content $i$ can be derived as follows.  
\begin{lemma}
\label{near-dist}
The \ac{PDF} of the distance to the nearest active provider for content $i$ is given by:
\begin{align}
\label{near-dist00}
 f_{H_i}(h_i) = &b_ip\bar{n} \int_{v_0=0}^{\infty} f_{V_0}(v_0) f_{H_i|V_0}(h_i|v_0) \times
 \nonumber  \\
 &{\rm exp}\Big(-b_ip\bar{n}\int_{0}^{h_i}f_{H|V_0}(h|v_0)\dd{h}\Big)\dd{v_0}, 
 \end{align}
where $f_{V_0}(v_0)=\mathrm{Rayleigh}(v_0,\sigma)$ is a Rayleigh \ac{PDF} of a scale parameter $\sigma$, which models the distance between the typical device and its cluster center; $f_{H|V_0}(h|v_0)=\mathrm{Rice} (h| v_0, \sigma)$ is a Rician \ac{PDF} of parameter $\sigma$ that models the distance $H=\lVert\boldsymbol{x}+\boldsymbol{y}\rVert$ between an intra-cluster device located at $\boldsymbol{y}$ relative to the cluster center and the origin $(0,0)$. 
\end{lemma}
\begin{proof}
Please see Appendix \ref{near-dist-proof} for the proof.
\end{proof}

In Fig.~\ref{near-serv-dist}, we verify the accuracy of the derived \ac{PDF} $ f_{H_i}(h_i)$ in (\ref{near-dist00}). The figure shows that the analytical expression matches the simulation well. Fig.~\ref{near-serv-dist} also shows that the distance to the nearest provider statistically decreases when $b_i$ increases. This is intuitive because it is more likely to have closer providers when they cache the  content with higher probabilities.  

It is worth highlighting that, when a desired content $i$ is downloaded from the nearest active provider, the distance \ac{PDF} $f_{H_i}(h_i)$ depends on the caching probability $\boldsymbol{b}$. This is different from the case when the  content is delivered by a randomly-selected provider within the same cluster. To illustrate, from the definition of Gaussian \ac{PPP}, the \ac{PDF} of the distance from the client to a randomly-selected provider is $f_R(r)=\frac{r}{2\sigma^2}e^{-\frac{r^2}{4\sigma^2}}$, which is, remarkably, independent of the caching probability $\boldsymbol{b}$. This leads to that, if a content is brought from a randomly-selected  provider, all contents are downloadable  with the same average rate irrespective of their caching probability. However, when downloading from nearest providers, this yields different rates for different content. This fact is crucial in the delay analysis, as we will discuss in Section \ref{del-anal}.



\begin{figure}  [!t]
\vspace{-0.4 cm}
\centering
{ \hspace*{-0.0 in}
  \subfigure[$b_i=1$]{\includegraphics[width=3.0 in]{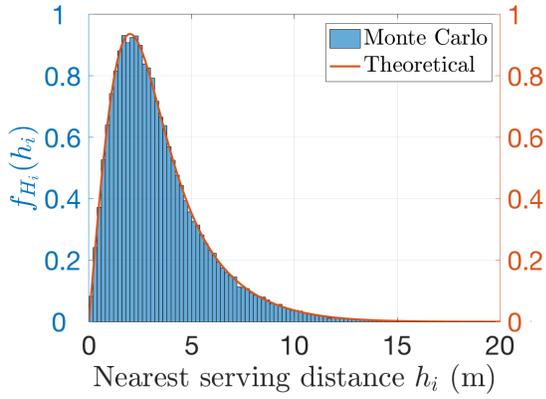}		 
\label{delay_compare}}
  \subfigure[$b_i=0.5$]{\includegraphics[width=3.0 in]{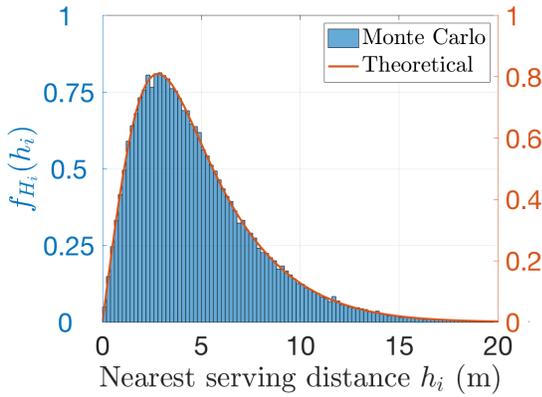}		 
\label{delay_compare}}
  }
\caption{The \ac{PDF} of the nearest serving distance is presented versus $h_i$ at different caching probabilities: ($\bar{n}=20$, $\sigma=\SI{5}{m}$, and $p=0.5$).}
\label{near-serv-dist}
\vspace{-0.6cm}
\end{figure}

Having characterized the nearest distance \ac{PDF} $f_{H_i}(h_i)$, next, we conduct the \ac{D2D} coverage probability analysis. The received power at the client device from the nearest content provider located at $\boldsymbol{y}_{0i}$ relative to the cluster center at $\boldsymbol{x}_0$ is given by 
\begin{align}
P_i &= P_d  g_{0i} \lVert \boldsymbol{x}_0+\boldsymbol{y}_{0i}\rVert^{-\alpha}= P_d  g_{0i} h_i^{-\alpha}	
\label{pwr}
\end{align}
where $g_{0i} \sim {\rm exp}(1)$ is the fading of the desired link. The typical device sees two types of interference, particularly, intra-and inter-cluster interference. The former is generated from active devices within the same cluster while the latter is produced from active devices in the remote clusters. The set of active devices in a remote cluster is denoted as $\mathcal{B}^p$, where $p$ refers to the access probability. Similarly, the set of active devices in the local cluster is denoted as $\mathcal{A}^p$. The inter-cluster interference is hence given by 
\begin{align}
I_{\Phi_p^{!}} &= \sum_{\boldsymbol{x} \in \Phi_p^{!}} \sum_{\boldsymbol{y} \in \mathcal{B}^p}   P_d g_{\boldsymbol{y}_{\boldsymbol{x}}}  \lVert \boldsymbol{x}+\boldsymbol{y}\rVert ^{-\alpha}
=  \sum_{\boldsymbol{x}\in \Phi_p^{!}} \sum_{\boldsymbol{y} \in \mathcal{B}^p}  P_d g_{u}  u^{-\alpha},
\end{align}
where $\Phi_p^{!}=\Phi_p \setminus \boldsymbol{x}_0$, $\boldsymbol{y}\in \mathcal{B}^p$ are the locations of interfering devices relative to their cluster center at $\boldsymbol{x}\in\Phi_p^{!}$, and $u = \lVert \boldsymbol{x}+\boldsymbol{y}\rVert $ are the inter-cluster interfering distances, see Fig. \ref{distance-delay-queue};  $g_{\boldsymbol{y}_{\boldsymbol{x}} } \sim {\rm exp}(1)$ is the fading of an interfering link, and $g_{u} = g_{\boldsymbol{y}_{\boldsymbol{x}} }$ for ease of notation. For the intra-cluster interference, we further divide $\mathcal{A}^p$ into two sets, namely, $\mathcal{A}_1^p$ and $\mathcal{A}_2^p$. We let $\mathcal{A}_1^p$ and $\mathcal{A}_2^p$ be the set of active devices farther, and closer than the distance $h_i$, respectively. Clearly, active devices belonging to $\mathcal{A}_2^p$ do not cache the desired content $i$. Hence, we have 
\begin{align}
I_{\Phi_c}  &= \sum_{\boldsymbol{y} \in \mathcal{A}_1^p}   P_d g_{\boldsymbol{y}_{\boldsymbol{x}_0}}  \lVert \boldsymbol{x}_0+\boldsymbol{y}\rVert ^{-\alpha} + 
\sum_{\boldsymbol{y} \in \mathcal{A}_2^p}   P_d g_{\boldsymbol{y}_{\boldsymbol{x}_0}}  \lVert \boldsymbol{x}_0+\boldsymbol{y}\rVert ^{-\alpha}
\nonumber \\
&=   \sum_{\boldsymbol{y} \in \mathcal{A}_1^p}  P_d g_{r}  r^{-\alpha}+
\sum_{\boldsymbol{y} \in \mathcal{A}_2^p}  P_d g_{r}  r^{-\alpha}, 
\end{align}
where $\boldsymbol{y}\in\mathcal{A}^p$ are the locations of intra-cluster interfering devices relative to their cluster center at $\boldsymbol{x}_0 \in \Phi_p$, and $r = \lVert \boldsymbol{x}_0+\boldsymbol{y}\rVert$ are the intra-cluster interfering distances, see Fig. \ref{distance-delay-queue}. We substitute $g_{r} = g_{\boldsymbol{y}_{\boldsymbol{x}_0}}$ for ease of notation, where $g_{\boldsymbol{y}_{\boldsymbol{x}_0}} \sim {\rm exp}(1)$. Hence, the received \ac{SIR} at the typical device for content $i$ will be:
\begin{equation}
\gamma_{h_i}= \frac{P}{I_{\Phi_p^{!}} + I_{\Phi_c}} = \frac{P_d  g_{0i} h_i^{-\alpha}}{I_{\Phi_p^{!}} + I_{\Phi_c}}.
\end{equation}

The probability that the typical device receives content $i$ via \ac{D2D} communications with an \ac{SIR} higher than a target threshold $\theta$, i.e., the rate coverage probability, can be obtained from:
\begin{align}
\Upsilon_i= \Pb(\gamma_{h_i}>\theta)&= \Eb \Big[ \Pb \Big(\frac{P_d  g_{0i} h_i^{-\alpha}}{I_{\Phi_p^{!}} + I_{\Phi_c}} > \theta \Big) \Big]	 
\nonumber  \\
&=
 \Eb \Big[ \Pb\Big( g_{0i} >  \frac{\theta h_i^{\alpha}}{P_d} [I_{\Phi_p^{!}} + I_{\Phi_c}]  \Big) \Big]
\nonumber  \\
\label{prob-R1-g-R0}
&\overset{(a)}{=} \mathbb{E}_{I_{\Phi_p^{!}},I_{\Phi_c},h_i}\Big[{\rm exp}\big(\frac{-\theta h_i^{\alpha}}{P_d}{[I_{\Phi_p^{!}} + I_{\Phi_c}] }\big)\Big]
\nonumber  \\
&\overset{(b)}{=}  \Eb_{h_i} \mathscr{L}_{I_{\Phi_p^{!}}}(s) \mathscr{L}_{I_{\Phi_c}} (s),
\end{align}			
where (a) follows from $g_{0i} \sim {\rm exp}(1)$, and (b) follows from the assumption that intra- and inter-cluster interference are independent and the Laplace transform of them, with $s=\frac{\theta h_i^{\alpha}}{P_d}$. 
Next, we derive the Laplace transform of intra- and inter-cluster interference. 
\begin{lemma}
\label{intra-cluster}
Laplace transform of the intra-cluster interference can be approximated by
\begin{align}   
\label{LT_intra}
         \mathscr{L}_{I_{\Phi_c} }(s) & \approx  
         {\rm exp}\Big(-p\bar{n}\bar{b_i}\int_{r=0}^{h_i}\frac{sf_R(r)}{s+ r^{\alpha}}\dd{r}\Big) \times
         \nonumber  \\
         & {\rm exp}\Big(-p\bar{n} \int_{r=h_i}^{\infty}\frac{sf_R(r)}{s+ r^{\alpha}}\dd{r}\Big),
\end{align} 
where $\bar{b_i}=1-b_i$, and $f_R(r) =\mathrm{Rayleigh}(r,\sqrt{2}\sigma)$ is a Rayleigh \ac{PDF} of a scale parameter $\sqrt{2}\sigma$, which models the distance between an intra-cluster interfering device to the typical device.  
\end{lemma}
\begin{proof}
Please see Appendix \ref{app-LT_intra}.
\end{proof}
The accuracy of the adopted approximation in Lemma \ref{intra-cluster} will be verified via simulations. Having characterized Laplace transform of the intra-cluster interference, next, we similarly obtain Laplace transform of the inter-cluster interference.
\begin{lemma}
\label{inter-cluster}
Laplace transform of the inter-cluster aggregate interference $I_{\Phi_p^{!}}$ is given by 
\begin{align}
 \label{LT_inter}
\mathscr{L}_{I_{\Phi_p^{!}}}(s) &= {\rm exp}\Big(-2\pi\lambda_p \int_{v=0}^{\infty}\Big(1 -  {\rm e}^{-p\bar{n} \varphi(s,v)}\Big)v\dd{v}\Big),
\end{align}
where $s=\frac{\theta h_i^{\alpha}}{P_d}$, $\varphi(s,v) = \int_{u=0}^{\infty}\frac{s}{s+ u^{\alpha}}f_U(u|v)\dd{u}$, and $f_U(u|v) = \mathrm{Rice} (u| v, \sigma)$ is a Rician \ac{PDF} of parameter $\sigma$ that models the distance $U=\lVert\boldsymbol{x}+\boldsymbol{y}\rVert$ between an interfering device at $\boldsymbol{y}$ relative to its cluster center at $\boldsymbol{x} \in \Phi_{p}$ and the origin $(0,0)$, conditioned on $V=\lVert\boldsymbol{x}\rVert=v$.
\end{lemma}
\begin{proof}
Please see Appendix \ref{app-LT_inter}.
\end{proof} 	

\begin{figure}[!t] 
\vspace{-0.4 cm}
	\begin{center}
		\includegraphics[width=3.5in]{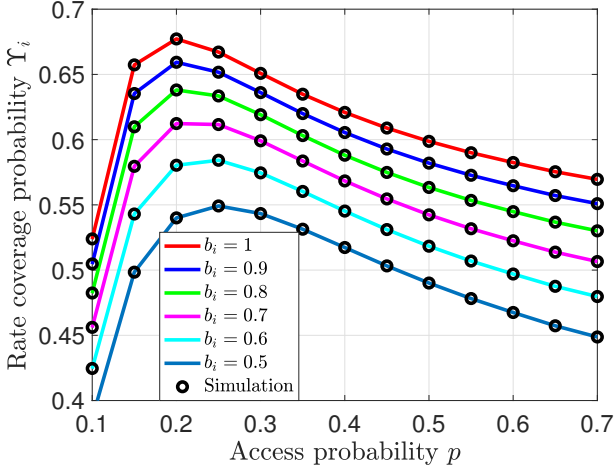}	
		\caption {The \ac{D2D} rate coverage probability $\Upsilon_i$ versus the access probability $p$ ($\lambda_p=\SI{50}{km^{-2}},\sigma=\SI{10}{m},\bar{n}=20,\theta=\SI{0}{dB}$).}
		\label{cov_prob_vs_sigma_new}
	\end{center}
\vspace{-0.7cm}
\end{figure} 

From Lemma \ref{intra-cluster} and  \ref{inter-cluster}, the rate coverage probability of content $i\in \mathcal{F}$ will be given by:
\begin{align}  
\label{file-rate-cov-prob}
 \Upsilon_i = \int_{h_i=0}^{\infty}\Pb(\gamma_{h_i}>\theta)f_{H_i}(h_i)\dd{h_i}. 
\end{align} 

In Fig.~\ref{cov_prob_vs_sigma_new}, we plot the \ac{D2D} rate coverage probability $\Upsilon_i$ against the channel access probability $p$ at different caching probability $b_i$. Fig.~\ref{cov_prob_vs_sigma_new} first verifies the accuracy of the approximation in Lemma \ref{intra-cluster}. Moreover, it is clear that the rate coverage probability for content $i$ increases with its caching probability $b_i$. This is because contents that are cached with higher probabilities are reachable with statistically shorter distances, and, subsequently, stronger received signal power. Moreover, Fig.~\ref{cov_prob_vs_sigma_new} also illustrates that $\Upsilon_i$ first increases with $p$ driven by the increasing probability of channel access. Then, $\Upsilon_i$ declines as $p$ increases due to the effect of adverse interference conditions. \emph{The access probability can be treated as an input parameter that reflects the interest of devices to participate in content sharing based on various parameters, e.g., their battery levels or cache sizes.}

\vspace{-0.2 cm}
\subsection{\ac{BS}-to-Device Coverage Probability}
If a content client attaches to the nearest BS from $\Phi_{b}$ to download a desired content, the \ac{BS}-to-device  coverage probability can be expressed similar to \cite{6042301} as: 
$\Upsilon_b =\frac{1}{{}_2 F_1(1,-\delta;1-\delta;-\theta)}$, 
where ${}_2 F_1(.)$ is the Gaussian hypergeometric function and $\delta = 2/\alpha$. 

Having obtained the rate coverage probabilities for \ac{D2D} and BS-to-device links, next, we characterize the average service delay and formulate the average service delay minimization problem. 

\vspace{-0.4cm}
\section{Delay Analysis and Minimization}
\label{del-anal}
In this section, we conduct the delay analysis and formulate the delay minimization problem. To do so, we start by characterizing the average service rates and corresponding service delays of the queues Q$_d$ and Q$_b$, which model requests served via D2D and BS-to-device communications, respectively.

\vspace{-0.2 cm}
\subsection{D2D Service Rate and Service Delay} 
Having explained, the rate coverage probability depends on its caching probability $b_i$, since $b_i$ impacts the distance distribution to the nearest active provider as given in (\ref{near-dist00}). Additionally, given that $C_i = W_d{\rm log}_{2}(1+ \theta)\Upsilon_i$ and $\mu_i=\frac{C_i}{\bar{S}}$, $i\in\mathcal{F}$, different files are downloadable with different rates via \ac{D2D} communications.\footnote{In \cite{8412262} and \cite{yang2016analysis}, different service rates are equivalent to different transmission modes of serving files whereas each mode has a different transmission rate. However, even though Q$_d$ models the \ac{D2D} transmission mode, i.e., only one transmission mode, different contents are downloadable with different rates based on the content availability.}  
 Similar to \cite{8412262} and \cite{yang2016analysis}, we assume that the content size follows an exponential distribution of mean $\overline{S}$, hence, the content service time also obeys an exponential distribution with means $\tau_i=\frac{1}{\mu_i}$. 
Moreover, $Q_d$ represents a \emph{\ac{MPSQ}} whose arrival rate and service rates are  $\zeta_d$  and $\mu_i$,  $i\in\mathcal{F}$, respectively \cite{collings1974queueing}. This is because different contents have different \ac{D2D} service rates, see, e.g.,  \cite{8412262} and \cite{yang2016analysis}. Given the Poisson arrival process $\zeta_d$, content $i$ also has a Poisson arrival whose arrival rate and mean inter-arrival time are $\zeta_i=q_i(1-e^{-b_ip\bar{n}})\zeta$ and $\frac{1}{\zeta_i}$, respectively. Hence, the \ac{D2D} aggregate arrival rate is given by $\zeta_d=\sum_{i=1}^{N_f} \zeta_i= \zeta\sum_{i=1}^{N_f}q_i(1-e^{-b_ip\bar{n}})$. 

The traffic intensity of a queue is defined as the ratio of mean service time to mean inter-arrival time. Let $\rho_d$ be the intensity of traffic served via \ac{D2D} communications, where 
\begin{align}
\rho_d = \sum_{i=1}^{N_f}\frac{1/ \mu_i}{1/\zeta_i} = \sum_{i=1}^{N_f}\frac{\zeta_i}{\mu_i} = \frac{\bar{S}}{W_d {\rm log}_2(1+\theta)}\sum_{i=1}^{N_f}\frac{\zeta_i}{\Upsilon_i}.
 \end{align}
  The stability condition then requires  that $\rho_d < 1$, otherwise, the overall delay will be infinite. In addition, the mean queue size for an \emph{\ac{MPSQ}} with traffic intensity $\rho_d$ can be approximated by $L_d =  \frac{\rho_d}{1 - \rho_d}$ \cite{collings1974queueing}. From \cite{Kleinrock}, the average service delay is the ratio of mean queue size $L_d$ over the arrival rate $\zeta_d$. Hence, we can express the \ac{D2D} service delay as:  
\begin{align}
T_d &=\frac{L_d}{\zeta_d} =  \frac{1}{\zeta_d} \frac{\rho_d}{1-\rho_d}  =  \frac{1}{\zeta_d} \frac{1}{\frac{1}{\rho_d}-1} 
= \frac{1}{\zeta_d}   \frac{1}{ \frac{W_d {\rm log}_2(1+\theta)}{\bar{S}\sum_{i=1}^{N_f}\frac{\zeta_i}{\Upsilon_i}} -1}
 \nonumber \\ 		
 \label{d2d-del}
 &= \frac{\zeta}{\zeta_d}   
 \Bigg(\frac{ \sum_{i=1}^{N_f}\frac{q_i(1-e^{-b_ip\bar{n}})}{\Upsilon_i}}{\frac{W_d}{\bar{S}}{\rm log}_2(1+\theta) - \zeta \sum_{i=1}^{N_f}\frac{q_i(1-e^{-b_ip\bar{n}})}{\Upsilon_i}}\Bigg).
 \end{align}
 
 \begin{remark}
From (\ref{d2d-del}), it readily follows that the \ac{D2D} service delay $T_d$ decreases as the rate coverage probability $\Upsilon_i$ increases. Accordingly, the average service delay $T$ will also improve with the improvement  of the rate coverage probability. The rate coverage probability $\Upsilon_i$ increases as either the density of clusters or the distance between devices decrease, which correspond to lower interference and higher signal levels, respectively. This  characteristic is revisited in Section \ref{num-result} when we study the effect of network geometry on the average service delay.
  \end{remark}

\subsection{BS-to-device Service Rate and Service Delay} 
Given the \ac{BS} coverage probability $\Upsilon_b$, the \ac{BS}-to-device average rate will be: $C_b =  W_b{\rm log}_2(1+\theta)\Upsilon_b$, and correspondingly, the service rate is $\mu_b=\frac{C_b}{\bar{S}} = \frac{W_b}{\overline{S}}{\rm log}_2(1+\theta)\Upsilon_b$. 
%
To characterize the average service delay, we must calculate the average number of clients attached to the same \ac{BS} and specify the scheduling policy. Under the assumption of one content client per cluster, the locations of clients form a \ac{PPP} with the same intensity $\lambda_p$ as the cluster centers. This is a direct result from the displacement theory of the \ac{PPP} \cite{haenggi2012stochastic}, where the clients can be considered as the cluster centers that are independently displaced from their locations. Hence, the average number of clients per \ac{BS} is $\eta=\frac{\lambda_p}{\lambda_b}$ \cite{6287527}. Moreover, we assume a \ac{FIFO} scheduling where each \ac{BS}  schedules the earliest arriving request in each time slot. Hence, all queues at a \ac{BS} could be considered as a single large queue whose arrival rate is $\eta \zeta_b$, see Fig.~\ref{distance-delay-queue}. The average service delay of \ac{BS}-served requests of the typical device equals the average service delay of the requests arriving at this large queue. 

The probability that more than two requests arrive simultaneously at the large queue is very small \cite{zhong2017heterogeneous}. Therefore, the arrival process at the large queue can be considered as a Poisson process whose aggregate arrival rate is $\eta\zeta_b$. Given the exponential distribution of content size, the traffic model at each BS is equivalent to a M/M/1 queueing system whose  mean queue size, arrival rate, and service rate are, respectively, $L_b$, $\eta\zeta_b$, and $\mu_b$. The traffic intensity is hence $\rho_b = \frac{\eta\zeta_b}{\mu_b}$, and, correspondingly, the mean queue size will be $L_b =   \frac{\rho_b}{1 - \rho_b}=\frac{\eta\zeta_b}{\mu_b - \eta\zeta_b}$ \cite{Kleinrock}. The average service delay per request for Q$_b$ of the typical device can be obtained from 
\begin{align}
T_b &=  \frac{L_b}{\eta\zeta_b}=\frac{1}{\mu_b - \eta\zeta_b} 
= \frac{1}{\frac{W_b}{\overline{S}}{\rm log}_2(1+\theta)\Upsilon_b - \eta\zeta_b}. 
\end{align} 
%
 %


Having discussed, the service rates of Q$_d$ and Q$_b$ depend on the traffic of the network. This is attributed to the fact that requests being served are susceptible to interference only if there are concurrently-served requests in other clusters (or by other BSs). Next, we conduct the delay analysis for a dominant system as widely-adopted in the literature, see, e.g., \cite{zhong2016stability,zhong2017heterogeneous}, and \cite{zhong2017toward}. Particularly, the dominant network is a fictitious system that is identical to the original system, except that terminals may choose to transmit even when their respective buffers are empty, i.e., they transmit a dummy packet. If both original and dummy systems started from the same initial state and fed with the same arrivals, then, the queues in the fictitious dominant system can never be shorter than the queues in the original system. Hence, the obtained average service delay for the dominant system represents an upper bound on that of the original system. Hence, in the sequel, the obtained average service delay for the dominant system represents an upper bound on that of the original system.
\vspace{-0.3 cm}
\subsection{Average Service Delay Minimization}
Next, we first characterize the conditions under which the network sustains its stability, i.e., the average service delay is bounded almost surely. Then, we obtain and minimize the average service delay.
\begin{proposition} {\rm The original system is stable if Q$_d$ and Q$_b$ are stable, and the sufficient conditions of Q$_d$ and Q$_b$ stability are, respectively,}
\begin{align}
\label{stable1}
\sum_{i=1}^{N_f}\frac{\zeta_i}{\Upsilon_i} &< \frac{W_d {\rm log}_2(1+\theta)}{\bar{S}},
\\
\label{stable2}
\zeta_{b} &<\frac{W_b{\rm log}_2[1+\theta]\Upsilon_b}{\eta \overline{S}}. 
\end{align}
\end{proposition}
\begin{proof}
Equations (\ref{stable1}) and (\ref{stable2}) implies, respectively,  that $\rho_d$ and $\rho_b$ are less than one, which guarantee the stability of Q$_d$ and Q$_b$ in the dominant system where all queues that have empty buffers make dummy transmissions. By Loynes' theorem \cite{loynes1962stability}, it follows that lim$_{t\rightarrow \infty}Q_i(t)< \infty$, $i\in\{d,b\}$, (almost surely) for all queues in the dominant network. 
\end{proof}

The weighted average service delay $T$ will be then expressed as: 
\begin{align}
T &=  \frac{\zeta_{d}T_d + \zeta_{b}T_b}{\zeta}
\nonumber \\ 
& 
=  
 \frac{ \sum_{i=1}^{N_f}\frac{q_i(1-e^{-b_ip\bar{n}})}{\Upsilon_i}}{\frac{W_d}{\bar{S}}{\rm log}_2(1+\theta) - \zeta \sum_{i=1}^{N_f}\frac{q_i(1-e^{-b_ip\bar{n}})}{\Upsilon_i}} 
   \nonumber  \\ 
   &  \quad  + \frac{\zeta_{b}}{\zeta} 
\frac{1}{\frac{W_b}{\overline{S}}{\rm log}_2(1+\theta)\Upsilon_b - \eta\zeta_b}
\nonumber \\ 
& = \frac{ \sum_{i=1}^{N_f}\frac{q_i(1-e^{-b_ip\bar{n}})}{\Upsilon_i}}{\frac{W_d}{\bar{S}}{\rm log}_2(1+\theta) - \zeta \sum_{i=1}^{N_f}\frac{q_i(1-e^{-b_ip\bar{n}})}{\Upsilon_i}}  
  \nonumber \\ 
  &
  \label{delay-T-eqn}
 \quad +  \frac{\sum_{i=1}^{N_f} q_ie^{-b_ip\bar{n}} }{\frac{W-W_d}{\overline{S}}{\rm log}_2(1+\theta)\Upsilon_b - \eta \zeta\sum_{i=1}^{N_f} q_ie^{-b_ip\bar{n}}}.
\end{align} 

Next, we minimize the average per request service delay $T$ by optimizing the bandwidth allocation and caching probability $\boldsymbol{b}$. It should be readily apparent that $\boldsymbol{b}$ determines the arrival rates $\zeta_{d}$ and $\zeta_{b}$ while the service rates $\mu_d$ and $\mu_b$ heavily depend on the  allocated bandwidth. We hence formulate the delay joint caching and bandwidth allocation minimization problem as 
\begin{align}
\label{optimize_eqn3}
\textbf{P1:} \quad \quad&\underset{\boldsymbol{b},W_d}{\text{min}} \quad T(\boldsymbol{b},W_d) \\
\label{const110}
&\textrm{s.t.}\quad \sum_{i=1}^{N_f}b_i = M, \quad \quad b_i \in [0,1] \quad \forall i \in \mathcal{F}, \\				
\label{const111}
&  0 \leq W_d \leq W,    \\
\label{stab1}
&\zeta\sum_{i=1}^{N_f}\frac{q_i(1-e^{-b_ip\bar{n}})}{\mu_i} <1,
\\
\label{stab2}
&\eta \zeta \sum_{i=1}^{N_f} q_ie^{-b_ip\bar{n}} < \mu_b, 
\end{align}
where constraints (\ref{stab1}) and (\ref{stab2}) are the stability conditions for Q$_d$ and Q$_b$, respectively. Next, we first show that $T(\boldsymbol{b},W_d)$ is a convex function of $W_d$ for a given caching probability $\boldsymbol{b}$. 
\begin{lemma}
\label{lemma-optimal-w-1}
For fixed $\boldsymbol{b}$, the objective function of \textbf{P1} in (\ref{optimize_eqn3}) is convex \ac{w.r.t.} $W_d$, and the optimal bandwidth allocation is given in a closed-form expression as 
\begin{align}
\label{optimal-w-1}
 W_d^* &= \frac{\zeta A + \sqrt{\frac{A}{B \Upsilon_b}}\big(WC\Upsilon_b - \eta\zeta B \big)}{\big(C +  C\Upsilon_b\sqrt{\frac{A}{B \Upsilon_b}}\big)}, 
 %
 %
 %
%
\end{align}
where $A=\sum_{i=1}^{N_f}\frac{q_i(1-e^{-b_ip\bar{n}})}{\Upsilon_i}$, $B=\sum_{i=1}^{N_f} q_ie^{-b_ip\bar{n}}$, and $C =\frac{{\rm log}_2(1+\theta)}{\bar{S}}$.
\end{lemma}
\begin{proof}
The sketch of the proof is found in Appendix \ref{BW-allocate}.
\end{proof}
\begin{figure}  [!t]
\vspace{-0.3 cm}
\centering
{ \hspace*{-0.0 in}
  \subfigure[Optimal bandwidth $\frac{W_d^*}{W}$ versus the popularity exponent]{\includegraphics[width=3.2in]{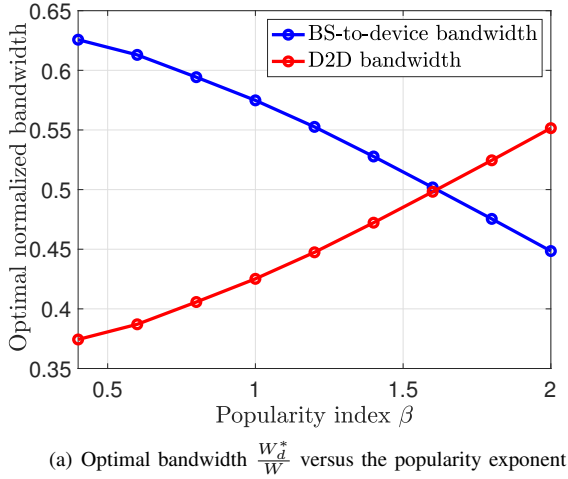}			
\label{BW_compare2}}
  }
 \subfigure[Optimal bandwidth $\frac{W_d^*}{W}$ versus the arrival rate $\zeta$]{\includegraphics[width=3.2in]{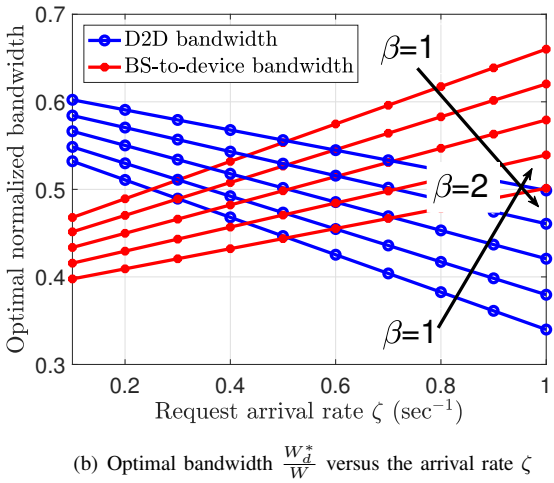}		
  \label{BW_compare22}}
\caption{The effect of various system parameters on the optimal bandwidth allocation $W_d^*$.}
\label{BW-analysis}
\vspace{-0.4 cm}
\end{figure}				
%
Before proceeding with the solution of \textbf{P1}, we first provide key insights on the performance of our proposed joint caching and bandwidth allocation scheme based on (\ref{optimal-w-1}). Firstly, from the first term in the numerator of (\ref{optimal-w-1}), it readily follows that the D2D optimal allocated bandwidth $W_d^*$ increases with the arrival rate of D2D communication-served requests. This is an interesting result since the system tends to cope up with such an increase in these \ac{D2D}-served requests by allocating more bandwidth to the \ac{D2D} communications. This behavior is shown in Fig.~\ref{BW-analysis} where we compare the bandwidth allocated to \ac{D2D} and base-to-device communications. In Fig.~\ref{BW_compare2}, we plot the optimal bandwidth $\frac{W_d^*}{W}$ versus the popularity exponent $\beta$. Fig.~\ref{BW_compare2} shows that the optimal allocated bandwidth $\frac{W_d^*}{W}$ monotonically increases with $\beta$, which can be interpreted as follows.  When $\beta$ increases, a smaller number of contents become highly demanded. These contents can be entirely cached and shared among the cluster devices. To cope with such a larger number of requests served via the \ac{D2D} communication, i.e., higher $\zeta_d$, the bandwidth $\frac{W_d^*}{W}$ needs to be increased. Fig.~\ref{BW_compare22} illustrates the effect of arrival rate $\zeta$ on the allocated bandwidth $W_d^*$. It is shown that $W_d^*$ tends to decrease as $\zeta$ increases, which can be interpreted as follows. Given that the aggregate arrival rate at each BS is $\eta\zeta_b$, when $\zeta$ increases, the average service delay can be minimized by balancing the loads and capacities of each communication mode. This can be attained by allocating more bandwidth to the \ac{BS}-to-device communication mode to cope up with the increasing load. \emph{This result reveals that it is vital to have joint load-aware resource allocation and content caching schemes to optimize the overall network performance.} 

Although the objective function of \textbf{P1} is convex \ac{w.r.t.} $W_d$, the coupling of the optimization variables $\boldsymbol{b}$ and $W_d$ makes \textbf{P1} a non-convex optimization problem. Therefore, \textbf{P1} cannot be solved directly using standard optimization methods. By applying the \ac{BCD} optimization technique, \textbf{P1} can be solved in an iterative manner similar to \cite{amer2022optimizing} as follows. First, for a given caching probability $\boldsymbol{b}$, we calculate $W_d^*$ from (\ref{optimal-w-1}). Afterwards, the obtained $W_d^*$ is used to update $\boldsymbol{b}$. Given $W_d^*$ from the bandwidth allocation subproblem, the caching probability subproblem can be formulated as
\begin{align}
\textbf{P2:} \quad \quad&\underset{\boldsymbol{b}}{\text{min}} \quad T(\boldsymbol{b},W_d^*) \\
&\textrm{s.t.}\quad  (\ref{const110}),  (\ref{stab1}), (\ref{stab2}).  \nonumber 
\end{align}

The caching probability subproblem \textbf{P2} is a sum of two fractional functions, where the first fraction is a prohibitively complex expression of $\boldsymbol{b}$, since $\Upsilon_i$ is an involved function of $\boldsymbol{b}$. This renders obtaining the optimal caching probability for \textbf{P2} very difficult. Instead, we seek a heuristic yet efficient algorithm to solve for a suboptimal caching solution for \textbf{P2}. Particularly, we use the interior point method, implemented in Mathematica, to obtain a suboptimal caching probability $\boldsymbol{b}$, as done in \cite{cache_schedule}. The entire proposed algorithm to solve \textbf{P1} is presented in Algorithm 1 and works as follows.\footnote{In the initialization of Algorithm 1, $T_0$ is first set to a large value, then, updated periodically based on the new $\boldsymbol{b}$ and $W_d$.} Firstly, we start with an initial caching probability $\boldsymbol{b}_0$ and allocated bandwidth $W_d=\frac{W}{2}$ to obtain a suboptimal caching solution based on the interior point method. Then, the obtained caching probability $\boldsymbol{b}$  is used to update the bandwidth allocation in (\ref{optimal-w-1}). The explained procedure, i.e., solving the two subproblems iteratively, is repeated until the value of \textbf{P1}'s objective function converges to a pre-specified accuracy. Importantly, the caching probability solution $\boldsymbol{b}$, given the optimal bandwidth $W_d^*$, depends on the initial value input to the interior point algorithm \cite{boyd2004convex}. We use the  Zipf caching as an initial point for this algorithm to obtain a suboptimal caching probability given the bandwidth calculated from (\ref{optimal-w-1}). 

\begin{algorithm}[!t]
    \SetKwInOut{Input}{Input}
    \SetKwInOut{Output}{Output}

    \Input{$W$, $N_f$, $M$, $\Upsilon_i$, $\Upsilon_b$, $\beta$, $\overline{S}$, $\theta$, $p$, $T_0$;}
    \textbf{Initialization}: {$\boldsymbol{b} \gets \boldsymbol{b}_0$, $W_{d} \gets \frac{W}{2}$\;
    ($T$) $\gets \text{\rm Eq. (\ref{delay-T-eqn})} \gets (\boldsymbol{b}_0,W_{d})$}\;
       \While{$T<T_0$}		
      {
      \tcc{Update $T_0$ with the calculated delay.}
      $T_0=T$\;
        ($\boldsymbol{b}$) $\gets \text{\rm interior point method}(\boldsymbol{b},W_d)$\;
        ($W_d$) $\gets \text{\rm Eq. (\ref{optimal-w-1})} \gets (\boldsymbol{b})$\;
        ($T$) $\gets \text{\rm Eq. (\ref{delay-T-eqn})} \gets (\boldsymbol{b},W_d)$\;

      }
      \Output{$\boldsymbol{b}$, $W_d$;}
 \caption{BCD algorithm for \textbf{P1}}
\end{algorithm}

%
%
%
%
%
%
%

\vspace{-0.0 cm}
\section{Numerical Results}
\label{num-result}
\begin{table}[!t]		
\vspace{-0.3 cm}
\caption{Simulation Parameters} 
\centering 
\scalebox{0.9}{
\begin{tabular}{|c| c| c| c| c | c|} 
\hline\hline 
Description & Parameter & Value\\ [0.5ex] 
\hline 
System bandwidth & $W$ & \SI{20}{\mega\hertz}\\ 
$\sir$ threshold&$\theta$&\SI{0}{\deci\bel}\\ 
Popularity exponent&$\beta$&0.5\\
Library size&$N_f$&10\\
Cache size per device&$M$&1\\
Path loss exponent&$\alpha$&4\\
Average number of devices per cluster&$\bar{n}$&20\\
Density of clusters&$\lambda_{p}$&$\SI{50}{km^{-2}}$\\
Total request arrival rate&$\zeta$&\SI{0.5}{request/sec}\\
 Average content size&$\overline{S}$&\SI{5}{Mbits} \\
Displacement standard deviation& $\sigma$ & \SI{10}{\metre}\\
 Density of BSs&$\lambda_{b}$& $\SI{10}{km^{-2}}$\\
\hline 
\end{tabular}}
\label{ch3:table:sim-parameter} 
\vspace{-0.5 cm}
\end{table}

At first, we validate the developed mathematical model via Monte Carlo simulations. Then we benchmark the proposed caching scheme against conventional caching schemes. Unless otherwise stated, the network parameters are selected as shown in Table \ref{ch3:table:sim-parameter}.\footnote{In Table \ref{ch3:table:sim-parameter}, we assume a unit cache size per device and a file library of relatively small size. These values, which are close to that used in \cite{8374852} and \cite{7995138}, are reasonable in the study of communication and caching aspects of D2D content delivery networks. Other works in the literature, e.g., \cite{6787081}, considered a much larger size of file library, however, their objective was to conduct the scaling analysis of caching networks.}

\begin{figure}[!tbh]
\vspace{-0.0 cm}
	\begin{center}
		\includegraphics[width=3.2in]{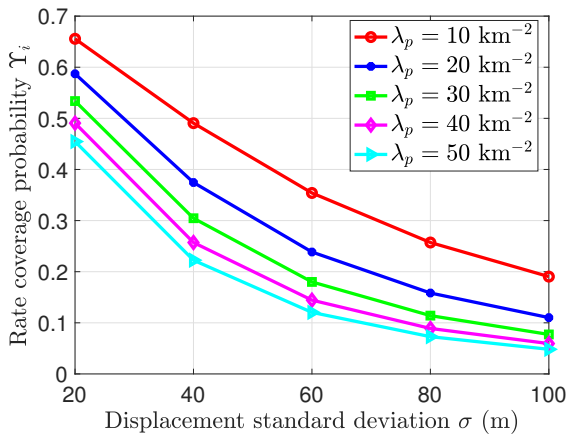}	
		\caption{The rate coverage probability $\Upsilon_i$ is plotted versus the network spatial parameters $\lambda_p$ and $\sigma$ ($b_i=1, p=0.2$).}
		\label{cache_size}
	\end{center}
\vspace{-0.3 cm}
\end{figure}

\begin{figure}  [!t]				
\centering
{ \hspace*{-0.0 in}
  \subfigure[Rate coverage probability $\Upsilon_i$ for different point processes, namely, \ac{TCP}, \ac{MCP}, and \ac{PPP}]{\includegraphics[width=3.2in]{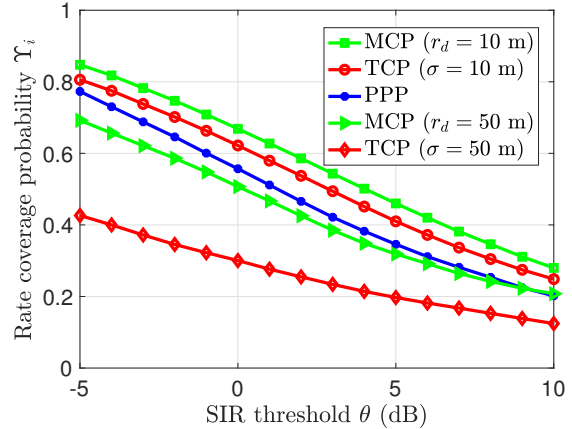}			
\label{rate-cov-analysis-ppp}}
}
\subfigure[Rate coverage probability $\Upsilon_i$ under different channel fading and different selections of the content provider]{\includegraphics[width=3.2in]{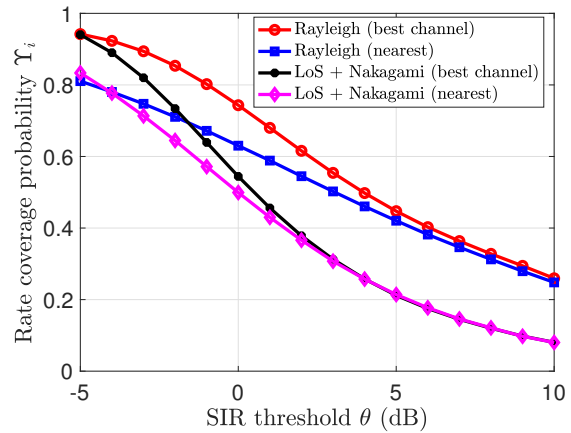}		
  \label{rate-cov-analysis-naka}}
\caption{Effects of the network topology and small-scale and  large-scale fading on the achievable rate coverage probability $\Upsilon_i$.}
\label{rate-cov-analysis-ppp-naka}
\vspace{-0.5cm}
\end{figure}	

\vspace{-0.4cm}
\subsection{\ac{D2D} Rate Coverage Probability} 
In Fig.~\ref{cache_size}, we plot the \ac{D2D} rate coverage probability $\Upsilon_i$ versus the displacement standard deviation $\sigma$ at different clusters' density $\lambda_p$. Fig.~\ref{cache_size} shows that $\Upsilon_i$ monotonically decreases with both $\sigma$ and $\lambda_p$. This is attributed to the fact that the  increase of $\sigma$ and $\lambda_p$ results in larger serving distance and shorter interfering distance, i.e., lower signal power due to the large-scale fading, and higher interference power received by the typical device, respectively. This explains the encountered degradation of $\Upsilon_i$ with $\sigma$ and $\lambda_p$. 

In Fig.~\ref{rate-cov-analysis-ppp-naka}, we investigate the effects of the network topology and small-scale fading parameters on the achievable performance. Firstly, Fig.~\ref{rate-cov-analysis-ppp} compares the achievable  rate coverage probability $\Upsilon_i$ of our proposed \ac{TCP} network to the \ac{PPP} and \ac{MCP} counterparts, which are assumed to operate as follows. For a fair comparison, we first set the \ac{PPP} intensity to be the same as the \ac{TCP}'s one of $\overline{n}\lambda_p$. Recall that $\overline{n}$ is the average number of devices per cluster and $\lambda_p$ is the density of clusters. Moreover, for the \ac{MCP}, the parent points are drawn from the same \ac{PPP} $\Phi_p$ as the \ac{TCP}. However, the cluster devices are  uniformly distributed in a ball of radius $r_d$ centered at the cluster center for the \ac{MCP} (with $r_d$ set similar to the standard deviation $\sigma$ for the \ac{TCP}). Fig.~8(a) first shows that depending on the distance between devices within the same cluster (i.e., $\sigma$ for \ac{TCP} and $r_d$ for \ac{MCP}), the achievable rate coverage probability $\Upsilon_i$ of the \ac{TCP} or \ac{MCP} might outperform that of \ac{PPP} or not. In particular, it is shown that if the distances between devices within the same cluster are sufficiently small, i.e., the smaller is $\sigma$ or $r_d$, the higher is $\Upsilon_i$ for the \ac{TCP}, and vice versa. Moreover, the achievable rate coverage probability for \ac{MCP} is shown to surpass that of the \ac{TCP} for the different setups of $\sigma$ (or $r_d$).


Secondly, in Fig.~\ref{rate-cov-analysis-naka}, we investigate the effect of small scale fading and \ac{LoS} communications on the achievable performance of our proposed network. Particularly, we compare the performance of the proposed clustered \ac{D2D} cache-enabled network, under the assumption of \ac{NLoS} communication and Rayleigh fading channels, to that of a network undergoing \ac{LoS}-dominated intra-cluster communication. In other words, we consider a variant of the channel model under which, devices within the same cluster  have direct \ac{LoS} communications, while those in different clusters experience \ac{NLoS} communication and Rayleigh fading channels. The inter-cluster communication channels are the same as described in Section \ref{netw-model}, i.e., characterized by exponentially distributed channel power gains with unit power, and path loss whose exponent is $\alpha$. However, the intra-cluster channel model follows the \ac{LoS} transmission approach of \cite{ding2016performance}, which is described as follows. The intra-cluster channels are characterized by both large-scale and small-scale fading. The large-scale fading is also represented by the path loss law, however, with a smaller path loss exponent $\alpha_l$, i.e., $\alpha_l<\alpha$. Moreover, the large scale fading is accompanied with a Nakagami small-scale fading whose fading parameter is $m>1$. The intra-cluster channel parameters are indicated in Table \ref{ch3:table:los-parameter}, which are similar to the \ac{LoS} channel parameters of \cite{8998329} and \cite{8756296}. From Fig.~\ref{rate-cov-analysis-naka}, it is clear that the existence of Nakagami fading along with \ac{LoS} communications among cluster devices substantially deteriorates the achievable rate coverage probability. This can be interpreted to that while Nakagami and \ac{LoS} communications lead to an improved (desired) signal level at the typical device, compared to \ac{NLoS} communications and Rayleigh fading channels, this also yields a stronger interference power and, accordingly, poorer rate coverage probability. Fig.~\ref{rate-cov-analysis-naka} also evaluates the \ac{D2D} rate coverage probability under two system setups, namely, nearest and channel-based content provider selections. For the latter, the \ac{CSI} is assumed to be known at the devices. Intuitively, the rate coverage probability is shown to improve under the channel-based content provider selection. However, such a gain is on the expense of high signaling and handover overheads. 

\begin{table}[!h]		
\vspace{-0.1 cm}
\caption{\ac{LoS} Simulation Parameters} 
\centering 
\scalebox{1.1}{
\begin{tabular}{|c| c| c| } 
\hline 
Description & Parameter & Value\\ [0.5ex] 
\hline 
Path loss exponent&$\alpha_l$&2.09\\ 
Nakagami fading parameter&$m$&$3$\\ 
\hline 
\end{tabular}}
\label{ch3:table:los-parameter} 
\end{table}

\begin{figure}  [!t]
\vspace{-0.3 cm}
\centering
{ \hspace*{-0.0 in}
  \subfigure[Optimized average service delay versus the popularity index.]{\includegraphics[width=3.2in]{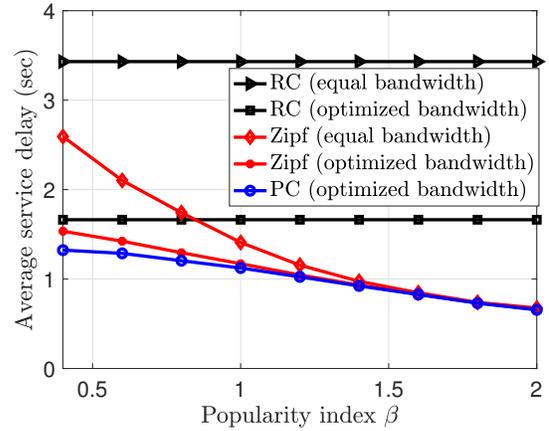}			
\label{delay_compare}}
}
{ \hspace*{+0.0 in}
\subfigure[Optimized average service delay $T$ versus the arrival rate $\zeta$.]{\includegraphics[width=3.2in]{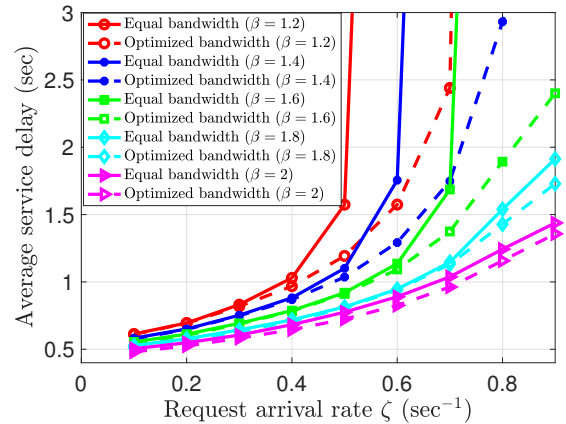}		
  \label{BW_compare}}
  }
\caption{Comparison of the proposed joint \ac{PC} and bandwidth allocation with other benchmark schemes ($\lambda_p=\SI{10}{km^{-2}}$, $\eta=5$).}
\label{delay-analysis}
\vspace{-0.5cm}
\end{figure}				

\vspace{-0.6 cm}
\subsection{Delay Results}		 
The performance of our proposed joint \ac{PC} and bandwidth allocation scheme is evaluated and compared with other benchmark schemes in Fig.~\ref{delay-analysis}. In particular, Fig.~\ref{delay_compare} compares the average service delay of our proposed scheme with \ac{RC} and Zipf caching schemes. For \ac{RC}, contents to be cached are uniformly chosen at random while for Zipf caching, the contents are chosen based on their popularity as in (\ref{zipf}). We first note that our proposed scheme significantly reduces the average service delay compared to the Zipf and \ac{RC} schemes under a  fixed bandwidth allocation (i.e., $W_d=W_b=W/2$). Particularly, for small $\beta$, the average service delay can be reduced by nearly $100\%$ ($350\%$) compared to  Zipf caching (\ac{RC}) under the fixed bandwidth allocation scheme. Moreover, under the optimized bandwidth allocation for \ac{RC} and Zipf, our proposed \ac{PC} still attains the best performance among all schemes. Meanwhile, except for the \ac{RC}, the average service delay monotonically decreases with the increase of $\beta$ since a smaller number of contents will undergo the highest demand.  
Fig.~\ref{BW_compare} manifests the effect of the request arrival rate $\zeta$ on the average service delay. We compare the performance of our proposed \ac{PC} under two schemes, namely, optimized and equally-divided  bandwidth allocation. Intuitively, the average service delay monotonically increases with $\zeta$ since larger request arrival rates result in a higher queuing delay for both Q$_d$ and Q$_b$. Moreover, we note that the system might become unstable, i.e., the average service delay goes to infinity, when $\zeta$ becomes considerably large. Interestingly, we see that the  proposed optimal bandwidth allocation scheme is more robust in terms of stability and can significantly reduce the average service delay. Particularly, the proposed scheme extends the stability region of the system in terms of request arrivals when compared to other schemes.

\begin{figure} [!t]	
\vspace{-0.1 cm}
\centering
{ \hspace*{-0.0 in}
\subfigure[Optimal bandwidth $\frac{W_d^*}{W}$ versus the displacement $\sigma$]{\includegraphics[width=3.2in]{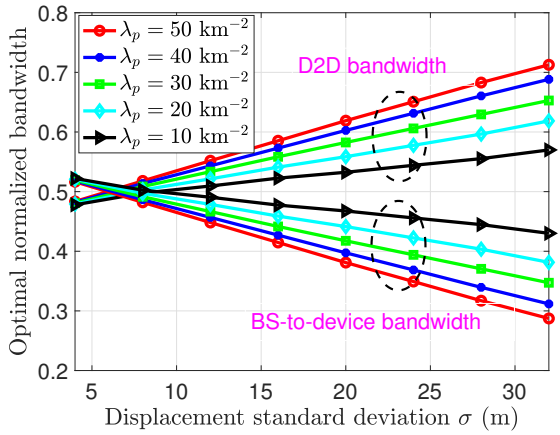}		 
  \label{optimal-w}}}
\subfigure[Optimized average service delay versus the displacement $\sigma$]{\includegraphics[width=3.2in]{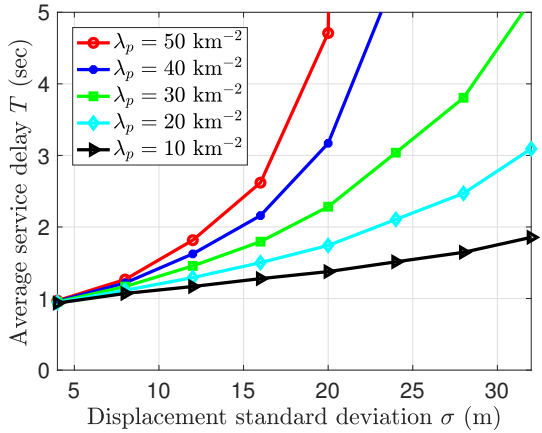}		 
  \label{av-delay}}
\caption{Effects of the network geometrical parameters on the average service delay and allocated bandwidth ($\beta=1.0$, $\eta=5$).} 
\vspace{-0.6cm}
\label{scaling}
\end{figure}

Fig.~\ref{scaling} shows the effect of the network geometry, e.g., clusters' density $\lambda_p$ and  displacement standard deviation $\sigma$, on the optimal allocated bandwidth $W_d^*$ and the average service delay. In Fig.~\ref{optimal-w}, we plot the optimal normalized bandwidth $\frac{W_d^*}{W}$ versus the displacement standard deviation $\sigma$ at different clusters' density $\lambda_p$. We first observe that the normalized optimal bandwidth tends to increase with both $\sigma$ and $\lambda_p$. This behavior can be understood in the light of (\ref{optimal-w-1}) as follows. Firstly, the rate coverage probabilities $\Upsilon_i$ decrease as $\sigma$ and $\lambda_p$ increase, as discussed above. Accordingly, the \ac{D2D} service rates $\mu_i=\frac{W_d{\rm log}_{2}(1+ \theta)\Upsilon_i}{\bar{S}}$, $i\in\mathcal{F}$, also decrease with the increase of $\sigma$ and $\lambda_p$. However, while the \ac{D2D} service rate $\mu_d$ tends to decrease with the decrease of $\Upsilon_i$, the optimal allocated bandwidth tends to increase to compensate for the service rate degradation, and eventually, minimizing the average service delay.

In Fig.~\ref{av-delay}, we plot the average service delay versus $\sigma$ at different clusters' density $\lambda_p$. We notice that the average service delay monotonically increases with  $\sigma$ and $\lambda_p$.  This can be interpreted similarly  to Fig.~\ref{cache_size} and Fig.~\ref{optimal-w} as follows. As $\sigma$ and $\lambda_p$ increase, the rate coverage probability $\Upsilon_i$ and, correspondingly, service rates $\mu_i$ will increase. This results in higher queuing and transmission delay in Q$_d$, and eventually, larger average service delay. 

 \begin{figure}[!t]
\vspace{-0.1 cm}
	\begin{center}
		\includegraphics[width=3.5in]{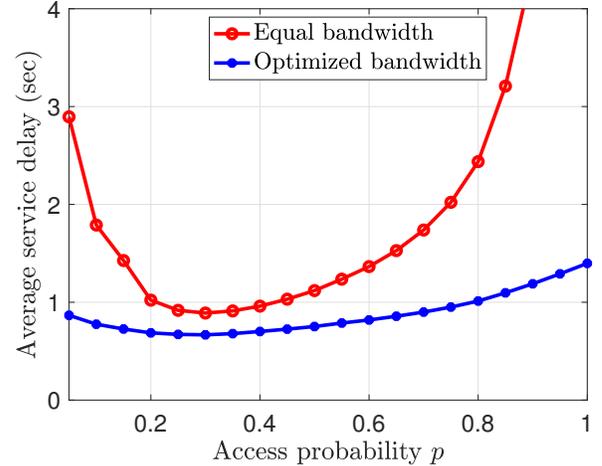}	
		\caption{The average service delay $T$ is plotted versus the access probability $p$ ($\theta=\SI{5}{dB}$, $\lambda_p=\SI{10}{km^{-2}}$, $\eta=5$).}
\label{delay_vs_p}
\end{center}
\vspace{-0.9cm}
\end{figure}

Finally, the effect of the access probability $p$ on the average service delay is investigated in Fig.~\ref{delay_vs_p}. The average service delay is plotted against $p$ for our proposed \ac{PC} under optimized and equally-divided  bandwidth allocation schemes. We can see that there exists an optimal $p$ that minimizes the average service delay. In essence, the value of $p$ yields an inherent tradeoff between the channel access opportunities and the effect of adverse interference. Clearly, our proposed bandwidth allocation attains the lowest average service delay and is seen to be more robust than the equally-divided  bandwidth scheme in terms of stability. Particularly, the proposed scheme extends the stability region of the system in terms of request arrivals when compared to other schemes. Having said, the channel access probability can be seen as an input to the system that reflects how individual devices will participate in the content delivery and sharing. 

 \vspace{-0.2 cm} 
\section{Conclusion}
In this paper, we have proposed a joint spectrum partitioning and content caching optimization framework that is leveraged to reduce the average service delay for clustered \ac{D2D} networks. Based on a developed spatiotemporal model, we have characterized the \ac{D2D} and \ac{BS}-to-device rate coverage probabilities and the  content arrival and service rates. We have then formulated the delay joint optimization problem whose decision variables are the content caching and bandwidth allocation. Employing a \ac{BCD} optimization technique, we have obtained the optimal allocated bandwidth in a closed-form expression, and a suboptimal caching scheme is also provided. Our results reveal that the joint optimization of spectrum partitioning and caching can effectively reduce the average service delay, e.g., by nearly $100\%$ and $350\%$, compared to the Zipf and uniform caching schemes under equal bandwidth allocations, respectively. Moreover, it is shown that the traffic load and content popularity play an important role in the design of resource allocation and content placement schemes.



\begin{appendices}
\section{Proof of Lemma \ref{near-dist}}
\label{near-dist-proof}

With reference to Fig.~\ref{distance-delay-queue}, the nearest distance $h_i$ is defined as the distance from the typical device at $(0,0)$ to its nearest active content provider belonging to the same cluster.  The \ac{PGF} of the number of active providers caching content $i$ within a ball $\textbf{b}(o, h_i)$ with radius $h_i$ and centered around the origin $o$ is
%
\begin{align}
& G_N(\vartheta)=  \Eb \left[ \vartheta^{\sum_{\boldsymbol{y}_{0i}\in\Phi_{ci}}\textbf{1}\{\lVert \boldsymbol{x}_0 + \boldsymbol{y}_{0i}\rVert < h_i\}}\right]   
\nonumber \\  
 &
 =  \Eb_{\Phi_c,\boldsymbol{x}_0} \prod_{\boldsymbol{y}_{0i}\in\Phi_{ci}} \left[ \vartheta^{\textbf{1}\{\lVert \boldsymbol{x}_0 + \boldsymbol{y}_{0i}\rVert < h_i\}}\right]    \nonumber \\
&\overset{(a)}{=}  \Eb_{\boldsymbol{x}_0} {\rm exp}\Big(-b_ip\overline{n}\int_{\mathbb{R}^2}(1 - \vartheta^{\textbf{1}\{\lVert \boldsymbol{x}_0 + \boldsymbol{y}_{0i}\rVert < h_i\}})f_{\boldsymbol{Y}_{0i}}(\boldsymbol{y}_{0i})d\boldsymbol{y}_{0i}\Big) \nonumber \\
&\overset{(b)}{=}  \Eb_{\boldsymbol{x}_0} {\rm exp}\Big(-b_ip\overline{n}\int_{\mathbb{R}^2}(1 - \vartheta^{\textbf{1}\{\lVert \boldsymbol{z}_0 \rVert < h_i\}})f_{\boldsymbol{Y}_{0i}}(\boldsymbol{z}_0-\boldsymbol{x}_0)d\boldsymbol{z}_0\Big), 				
 \end{align}
where $\textbf{1}\{.\}$ is the indicator function, and $\boldsymbol{x}_0\in \R^2$ is a \ac{RV} modeling the location of representative cluster's center relative to the origin $o$, with a realization $\boldsymbol{X}_0=\boldsymbol{x}_0$. (a) follows from the \ac{PGFL} of the \ac{PPP} $\Phi_{ci}$ along with its intensity function $\lambda_{ci}(y)=b_ip\overline{n} f_{\boldsymbol{Y}_{0i}}(\boldsymbol{y}_{0i})$, and (b) follows from change of variables $\boldsymbol{z}_0 = \boldsymbol{x}_0+\boldsymbol{y}_{0i}$. By converting Cartesian coordinates to polar coordinates with $h= \lVert\boldsymbol{x}_0+\boldsymbol{y}_{0i}\rVert = \lVert \boldsymbol{z}_0 \rVert$, we get
\begin{align}
&G_N(\vartheta)=  \Eb_{V_0} e^{-b_ip\overline{n}\int_{h=0}^{\infty}(1 - \vartheta^{\textbf{1}\{h < h_i\}})f_{H|V_0}(h|v_0)\dd{h}} 
\nonumber \\
&
\overset{(c)}{=} \Eb_{V_0} e^{-b_ip\overline{n}\int_{h=0}^{h_i}(1 - \vartheta)f_{H|V_0}(h|v_0)\dd{h}} 
 \overset{(d)}{=}  
\nonumber \\
 &
\int_{v_0=0}^{\infty}f_{V_0}(v_0)  {\rm exp}\Big(-b_ip\overline{n}\int_{h=0}^{h_i}(1 - \vartheta)f_{H|V_0}(h|v_0)\dd{h}\Big)\dd{v_0},
 \end{align}
 where $V_0\in \R$ is a \ac{RV} modeling the distance from representative cluster's center to the origin $o$, with a realization $V_0=v_0=\lVert \boldsymbol{x}_0\rVert$. (c) follows from the definition of the indicator function 
$\textbf{1}\{h < h_i\}$, and (d) follows from unconditioing over $v_0$. 

Given the intensity $\lambda_{ci}(y)$ of the Gaussian \ac{PPP} $\Phi_{ci}$, following \cite{7792210}, we can obtain the \ac{CDF} of the distance to the nearest content provider from:
\begin{align}
\label{near-cdf}
F_{H_i}&(h_i) =1 - G_N(0)= 
\nonumber \\
&1 - \int_{v_0=0}^{\infty} f_{V_0}(v_0) {\rm exp}\Big(-b_ip\bar{n}\int_{0}^{h_i}f_{H|V_0}(h|v_0)\dd{r}\Big)\dd{v_0},		
\end{align} 		
where $f_{V_0}(v_0) =\mathrm{Rayleigh}(v_0,\sigma)$ is a Rayleigh \ac{PDF} of a scale parameter $\sigma$, which models the distance $V_0=\lVert\boldsymbol{\boldsymbol{x}_0}\rVert$ between the client  device at the origin and the cluster center at $\boldsymbol{x}_0$, see Fig.~\ref{distance-delay-queue}. Similarly, $f_{R|V_0}(r|v_0)=\mathrm{Rice} (r;v_0,\sigma)$ is a Rician \ac{PDF} of a shape parameter $\frac{v_0}{2\sigma^2}$, which  models the distance $R=\lVert\boldsymbol{x}_0+\boldsymbol{y}\rVert$ between an intra-cluster active device at $\boldsymbol{y}$ relative to the cluster center at $\boldsymbol{x}_0 \in \Phi_{p}$ and the origin $(0,0)$, conditioned on $V_0=\lVert\boldsymbol{x}_0\rVert=v_0$. 
Applying Leibniz integral rule, which states that:
\begin{align}
& \frac{\dd}{\dd{x}}\Big(\int_{a(x)}^{b(x)}f(x,t)\Big)=
 f(x,b(x)) \frac{\dd}{\dd{x}}b(x) - f(x,a(x)) \frac{\dd}{\dd{x}}a(x)
 \nonumber \\
  &+ \int_{a(x)}^{b(x)}\frac{\partial}{\partial x}f(x,t),
 \end{align}
we obtain the nearest distance \ac{PDF} as
\begin{align}
&f_{H_i}(h_i)= \frac{\partial }{\partial h_i} \big( F_{H_i}(h_i) \big)= 
 \nonumber \\
 & -\frac{\partial}{\partial h_i} 
 \Big( \int_{v_0=0}^{\infty} f_{V_0}(v_0) {\rm exp}\Big(-b_i\bar{n}\int_{0}^{h_i}f_{H|V_0}(h|v_0)\dd{h}\Big)\dd{v_0}
 \Big)
\nonumber \\
&=  -\int_{v_0=0}^{\infty} f_{V_0}(v_0) \frac{\partial}{\partial h_i} \times
\nonumber \\
&
 {\rm exp}\Big(-b_i\bar{n}\int_{0}^{h_i}f_{H|V_0}(h|v_0)\dd{h}\Big)\dd{v_0}
\nonumber 
 \end{align}
 \begin{align}
&=  \int_{v_0=0}^{\infty} f_{V_0}(v_0)  \Big(b_i\bar{n}\frac{\partial}{\partial h_i}\int_{0}^{h_i}f_{H|V_0}(h|v_0)\dd{r}\Big) \times 
\nonumber \\
& {\rm exp}\Big(-b_i\bar{n}\int_{0}^{h_i}f_{H|V_0}(h|v_0)\dd{h}\Big)\dd{v_0}
\nonumber \\
\label{Leibniz}
&= b_ip\bar{n} \int_{v_0=0}^{\infty} f_{V_0}(v_0) f_{H_i|V_0}(h_i|v_0) \times 
\nonumber \\
&
{\rm exp}\Big(-b_ip\bar{n}\int_{0}^{h_i}f_{H|V_0}(h|v_0)\dd{h}\Big)\dd{v_0}.
 \end{align}

\section {Proof of lemma \ref{intra-cluster}}
\label{app-LT_intra}
Under the proposed channel access scheme, the set $\mathcal{A}_1^p$ forms a Gaussian \ac{PPP} of intensity $p\lambda_{c}(y)$, while $\mathcal{A}_2^p$ forms a Gaussian \ac{PPP} of intensity $p\bar{b_i}\lambda_{c}(y)$, where $b_i=1-\bar{b_i}$. With this in mind, Laplace transform of the intra-cluster interference $I_{\Phi_c}$ is obtained as follows: given the distance $v_0$ from the cluster center to the origin (see Fig.~\ref{distance-delay-queue}), we have $\mathscr{L}_{I_{\Phi_c} }(s|v_0)=$  
\begin{align}
&\mathbb{E} \Bigg[e^{-s \Big(\sum_{ \boldsymbol{y} \in \mathcal{A}_1^p}  g_{\boldsymbol{y}_{\boldsymbol{x}_0}}  \lVert \boldsymbol{x}_0  + \boldsymbol{y}\rVert^{-\alpha} + \sum_{ \boldsymbol{y} \in \mathcal{A}_2^p}  g_{\boldsymbol{y}_{\boldsymbol{x}_0}}  \lVert \boldsymbol{x}_0  + \boldsymbol{y}\rVert^{-\alpha} \Big)} \Bigg] \nonumber 
 \\ 
   &=  \mathbb{E}_{\Phi_{ci},g_{\boldsymbol{y}_{\boldsymbol{x}_0}}} 
   \prod_{ \boldsymbol{y} \in\mathcal{A}_1^p}  e^{-s g_{\boldsymbol{y}_{\boldsymbol{x}_0}}   \lVert \boldsymbol{x}_0  + \boldsymbol{y}\rVert^{-\alpha}}  
      \prod_{ \boldsymbol{y} \in\mathcal{A}_2^p}  e^{-s g_{\boldsymbol{y}_{\boldsymbol{x}_0}}   \lVert \boldsymbol{x}_0  + \boldsymbol{y}\rVert^{-\alpha}}  
      \nonumber \\
&=  \mathbb{E}_{\Phi_{ci}} \prod_{ \boldsymbol{y} \in\mathcal{A}^p_1}  \mathbb{E}_{g_{\boldsymbol{y}_{\boldsymbol{x}_0}}} e^{-s g_{\boldsymbol{y}_{\boldsymbol{x}_0}}   \lVert \boldsymbol{x}_0  + \boldsymbol{y}\rVert^{-\alpha}}  \times     
\nonumber  \\
&
\quad\quad\quad  \quad \quad \quad  \quad \quad \quad \prod_{ \boldsymbol{y} \in\mathcal{A}_2^p} \mathbb{E}_{g_{\boldsymbol{y}_{\boldsymbol{x}_0}}} e^{-s g_{\boldsymbol{y}_{\boldsymbol{x}_0}}   \lVert \boldsymbol{x}_0  + \boldsymbol{y}\rVert^{-\alpha}} 
\nonumber  \\
    &\overset{(a)}{=} \mathbb{E}_{\Phi_{ci}} \prod_{ \boldsymbol{y} \in\mathcal{A}_1^p} \frac{1}{1+s \lVert \boldsymbol{x}_0  + \boldsymbol{y}\rVert^{-\alpha}} 
  \mathbb{E}_{\Phi_{ci}} \prod_{ \boldsymbol{y} \in\mathcal{A}_2^p} \frac{1}{1+s \lVert \boldsymbol{x}_0  + \boldsymbol{y}\rVert^{-\alpha}}
        \nonumber \\
 & \overset{(b)}{=}  e^{-p\bar{n} \int_{\boldsymbol{y}\in\mathcal{A}_1^p}\big(1 - \frac{1}{1+s \lVert \boldsymbol{x}_0  + \boldsymbol{y}\rVert^{-\alpha}}\big)f_{\boldsymbol{Y}}(\boldsymbol{y})\dd{\boldsymbol{y}}} \times 
  \nonumber	
     \\
          &
  \quad \quad e^{-p\bar{b_i}\bar{n} \int_{\boldsymbol{y}\in\mathcal{A}_2^p}\big(1 - \frac{1}{1+s \lVert \boldsymbol{x}_0  + \boldsymbol{y}\rVert^{-\alpha}}\big)f_{\boldsymbol{Y}}(\boldsymbol{y})\dd{\boldsymbol{y}}},	 
     \nonumber	
\end{align}
where (a) follows from the Rayleigh fading assumption, and (b)  follows from the \ac{PGFL} of the Gaussian \ac{PPP} $\Phi_{ci}$ \cite{haenggi2012stochastic}. By changing the variables $\boldsymbol{z_0} = \boldsymbol{x}_0  + \boldsymbol{y}$ with $\dd \boldsymbol{z_0} = \dd{\boldsymbol{y}}$, we get  
\begin{align}
& \mathscr{L}_{I_{\Phi_c} }(s|v_0)  
        \overset{}{=} e^{-p\bar{n} \int_{\boldsymbol{y}\in\mathcal{A}_1^p}\big(1 - \frac{1}{1+s \lVert \boldsymbol{z_0}\rVert^{-\alpha}}\big)f_{\boldsymbol{Y}}(\boldsymbol{z_0}-\boldsymbol{x}_0)\dd{\boldsymbol{z_0}}}     \times 
\nonumber  \\      
 &
        e^{-p\bar{b_i}\bar{n} \int_{\boldsymbol{y}\in\mathcal{A}_2^p}\big(1 - \frac{1}{1+s \lVert \boldsymbol{z_0}\rVert^{-\alpha}}\big)f_{\boldsymbol{Y}}(\boldsymbol{z_0}-\boldsymbol{x}_0)\dd{\boldsymbol{z_0}}}
    \nonumber  \\      
 &\overset{(c)}{=} {\rm exp}\Big(-p\bar{n} \int_{r=h_i}^{\infty}\frac{sf_R(r|v_0)}{s+  r^{\alpha}}\dd{r}\Big) \times 
\nonumber  \\      
 &
  {\rm exp}\Big(-p \bar{b_i} \bar{n} \int_{r=0}^{h_i} \frac{sf_R(r|v_0)}{s+  r^{\alpha}}\dd{r}\Big)		
\nonumber 	
 \end{align}
where (c) follows from converting the cartesian coordinates to polar coordinates, with $r=\lVert \boldsymbol{z_0}\rVert$. To clarify how in (c) the normal distribution $f_{\boldsymbol{Y}}(\boldsymbol{z}_0-\boldsymbol{x}_0)$ is converted to the Rice distribution $f_R(r|v_0)$, recall first that the representative cluster is centered at $\boldsymbol{x}_0$, with a distance $v_0=\lVert \boldsymbol{x}_0\rVert$ from the origin. Further, intra-cluster interfering devices have their coordinates relative to $\boldsymbol{x}_0$ chosen independently from a Gaussian distribution of standard deviation $\sigma$. Then, by definition, the distance $r$ from an interfering device to the origin has a Rician \ac{PDF} $f_R(r|v_0)$ of a shape parameter $\frac{v_0}{2\sigma^2}$. Neglecting the correlation of the intra-cluster interfering distances, i.e., the common part $\boldsymbol{x}_0$ in $r=\lVert \boldsymbol{x}_0  + \boldsymbol{y}\rVert$, $ \boldsymbol{y} \in \mathcal{A}^p$, we can obtain a simple yet accurate expression of the Laplace transform as follows. Since both the client and interfering devices have their locations drawn from a Gaussian distribution of variance $\sigma^2$ relative to cluster center, then, every intra-cluster interfering distance has a Rayleigh \ac{PDF} of parameter $\sqrt{2}\sigma$, which yields 
\begin{align}
         \mathscr{L}_{I_{\Phi_c} }(s) &\approx  {\rm exp}\Big(-p\bar{n} \int_{r=h_i}^{\infty}\frac{s}{s+ r^{\alpha}}f_R(r)\dd{r}\Big) \times 
         \nonumber \\
& {\rm exp}\Big(-p\bar{b_i}\bar{n} \int_{r=0}^{h_i}\frac{s}{s+ r^{\alpha}}f_R(r)\dd{r}\Big),
\end{align}
when the correlation is neglected. This completes the proof.

\section{Proof of lemma \ref{inter-cluster}}
\label{app-LT_inter}
Laplace transform of the inter-cluster interference $I_{\Phi_p^{!}}$ can be evaluated as
\begin{align}
&\mathscr{L}_{I_{\Phi_p^{!}}}(s) = \mathbb{E} \Bigg[e^{-s \sum_{\Phi_p^{!}} \sum_{ \boldsymbol{y} \in \mathcal{B}^p}  g_{\boldsymbol{y}_{x}}  \lVert \boldsymbol{x}  +  \boldsymbol{y}\rVert^{-\alpha}} \Bigg] \nonumber \\ 
   &= \mathbb{E}_{\Phi_p} \Bigg[\prod_{\Phi_p^{!}} \mathbb{E}_{\Phi_{ci},g_{\boldsymbol{y}_{x}}} \prod_{ \boldsymbol{y} \in \mathcal{B}^p}  e^{-s g_{\boldsymbol{y}_{x}}   \lVert \boldsymbol{x}  +  \boldsymbol{y}\rVert^{-\alpha}} \Bigg]  
  \nonumber \\
   & = \mathbb{E}_{\Phi_p} \Bigg[\prod_{\Phi_p^{!}} \mathbb{E}_{\Phi_{ci}} \prod_{ \boldsymbol{y} \in \mathcal{B}^p} \mathbb{E}_{g_{\boldsymbol{y}_{x}}}  e^{-s g_{\boldsymbol{y}_{x}}   \lVert \boldsymbol{x}  +  \boldsymbol{y}\rVert^{-\alpha}} \Bigg]  \nonumber \\
    &\overset{(a)}{=} \mathbb{E}_{\Phi_p} \Bigg[\prod_{\Phi_p^{!}} \mathbb{E}_{\Phi_{ci}} \prod_{ \boldsymbol{y} \in \mathcal{B}^p} \frac{1}{1+s \lVert \boldsymbol{x}  +  \boldsymbol{y}\rVert^{-\alpha}} \Bigg]  
     \nonumber \\  &
    \label{LT_inter}
    \overset{(b)}{=} \mathbb{E}_{\Phi_p} \prod_{\Phi_p^{!}} {\rm exp}\Big(-p\bar{n} \int_{\mathbb{R}^2}\Big(1 - \frac{1}{1+s \lVert \boldsymbol{x}  +  \boldsymbol{y}\rVert^{-\alpha}}\Big)f_{\boldsymbol{Y}}(\boldsymbol{y})\dd{\boldsymbol{y}}\Big) 	 
   \nonumber	    \\
          &\overset{(c)}{=}  {\rm exp}\Bigg(-\lambda_p \int_{\mathbb{R}^2}\Big(1 -  {\rm exp}\Big(-p\bar{n} \int_{\mathbb{R}^2}\Big(1 - 
  \nonumber	
     \\
          &        
          \frac{1}{1+s \lVert \boldsymbol{x}  +  \boldsymbol{y}\rVert^{-\alpha}}\Big)f_{\boldsymbol{Y}}(\boldsymbol{y})\dd{\boldsymbol{y}}\Big)\dd{\boldsymbol{x}}\Bigg), 	             	
\end{align}
where (a) follows from the Rayleigh fading assumption, 
(b) follows from the \ac{PGFL} of the Gaussian \ac{PPP} $\Phi_{ci}$, and (c) follows from the \ac{PGFL} of the parent \ac{PPP} $\Phi_p$ \cite{haenggi2012stochastic}. By changing of variables $\boldsymbol{z} = \boldsymbol{x}  +  \boldsymbol{y}$ with $\dd \boldsymbol{z} = \dd \boldsymbol{y}$, we get
\begin{align}
\mathscr{L}_{I_{\Phi_p^{!}}}(s) &\overset{}{=}  {\rm exp}\Bigg(-\lambda_p \int_{\mathbb{R}^2}\Big(1 -  {\rm exp}\Big(-p\bar{n} \int_{\mathbb{R}^2}\Big(1 - 
\nonumber	
     \\
          &
          \frac{1}{1+s \lVert \boldsymbol{z}\rVert^{-\alpha}}\Big)f_{\boldsymbol{Y}}(\boldsymbol{z}-\boldsymbol{x})\dd{\boldsymbol{y}}\Big)\dd{\boldsymbol{x}}\Bigg)			 \nonumber	
          \end{align}
\begin{align}
          &\overset{(d)}{=}  {\rm exp}\Bigg(-2\pi\lambda_p \int_{v=0}^{\infty}\Big(1 -  {\rm exp}\Big(-p\bar{n} \int_{u=0}^{\infty}\Big(1 -
          \nonumber	
     \\
          &
           \frac{1}{1+s  u^{-\alpha}}\Big)f_U(u|v)\dd{u}\Big)v\dd{v}\Bigg)		 
          \nonumber \\
         &=  {\rm exp}\Big(-2\pi\lambda_p \int_{v=0}^{\infty}\Big(1 -  {\rm e}^{-p\bar{n} \varphi(s,v)}\Big)v\dd{v}\Big),	
\end{align}
where $\varphi(s,v) = \int_{u=0}^{\infty}\frac{s}{s+ u^{\alpha}}f_U(u|v)\dd{u}$; (d) follows from converting the cartesian coordinates to polar coordinates with $u=\lVert \boldsymbol{z}\rVert$. This completes the proof. 

\section {Proof of lemma \ref{lemma-optimal-w-1}}
\label{BW-allocate}
The weighted average service delay can be rewritten as 
\begin{align}
\label{optimize_eqn3_p1}
 T(\boldsymbol{b},W_d) &=
A \Big(W_d C - \zeta A \Big)^{-1}
+ B \Big((W-W_d)C\Upsilon_b - \eta\zeta B \Big)^{-1}, 
\end{align}
and the delay first and second derivatives with respect to $W_d$ will be given, respectively, by 
\begin{align}
 \frac{\partial T(\boldsymbol{b},W_d)}{\partial W_d} &= 
 -AC\Big(W_d C - \zeta A\Big)^{-2}   + 
 \nonumber	
     \\
          &
 B C \Upsilon_b \Big((W-W_d)C\Upsilon_b - \eta\zeta B \Big)^{-2}, 
\end{align}
and  
\begin{align}
  \frac{\partial^2 T(\boldsymbol{b},W_d)}{\partial W_d^2}&=
  2AC^2\Big(W_d C - \zeta A\Big)^{-3}   + 
  \nonumber	
     \\
          &
 2 B (C\Upsilon_b)^2 \Big((W-W_d)C\Upsilon_b - \eta\zeta B \Big)^{-3}.		
\end{align}
The stability conditions of Q$_d$ and Q$_b$, respectively, require that $W_d C>\zeta A$, and $(W-W_d)C\Upsilon_b > \eta\zeta B$. Hence, $\frac{\partial^2 T(\boldsymbol{b},W_d)}{\partial W_d^2}  > 0$, and, correspondingly,  the objective function is convex \ac{w.r.t.} $W_d$. The optimal bandwidth allocation can be directly obtained by from $\frac{\partial T(\boldsymbol{b},W_d)}{\partial W_d} =0$, which yields:
\begin{align}
 & \frac{AC}{\Big(W_d C - \zeta A\Big)^{2}}   = 
 \frac{B C \Upsilon_b}{\Big((W-W_d)C\Upsilon_b - \eta\zeta B \Big)^{2}}
 \nonumber \\
 \label{rearrange}
  & \sqrt{A}\Big((W-W_d)C\Upsilon_b - \eta\zeta B \Big)   = 
 \sqrt{B \Upsilon_b} \Big(W_d C - \zeta A\Big)
  \end{align}
 By rearranging (\ref{rearrange}), we get the optimal spectrum partitioning from 
 \begin{align}
 W_d^* &= \frac{\zeta A + \sqrt{\frac{A}{B \Upsilon_b}}\big(WC\Upsilon_b - \eta\zeta B \big)}{\big(C +  C\Upsilon_b\sqrt{\frac{A}{B \Upsilon_b}}\big)}. 
\end{align}
 This completes the proof. 
\textcolor{orange}{\cite{baza5,baza2,baza3,baza6,baza4,baza1,baza7,baza13,baza9,baza11,baza10,baza12,baza8,yilmaz2019expansion,yilmaz2}}

\end{appendices}
\bibliographystyle{IEEEtran}

\bibliography{bibliography}

\end{document}